\documentclass[a4paper, onecolumn, accepted=2020-12-22]{quantumarticle}
\pdfoutput=1
\usepackage[utf8]{inputenc}
\usepackage{scrextend}
\usepackage[english]{babel}
\usepackage{paralist}
\usepackage{amsmath}
\usepackage{amssymb}
\usepackage{multicol}
\usepackage{enumitem}
\usepackage{amsfonts}
\usepackage{amsthm}
\usepackage{mathrsfs}
\usepackage{graphicx}
\usepackage{stmaryrd}
\usepackage{mathtools}
\usepackage{hyperref}
\usepackage{microtype}
\usepackage{bbm}
\usepackage[numbers,sort&compress]{natbib}

\DeclarePairedDelimiter{\floor}{\lfloor}{\rfloor}

\DeclarePairedDelimiter{\norm}{\lVert}{\rVert}

\newcommand{\R}{\mathbb{R}}

\newcommand{\N}{\mathbb{N}}

\newcommand{\id}{\text{id}}
\newcommand\Define[1]{\textbf{#1}}

\DeclareMathOperator{\HS}{HS}

\newcommand{\mult}{\mathop{\circ}}
\newcommand{\commu}{\mathrel{\lvert}}

\newcommand{\bigovee}{\mathop{\vphantom{\sum}\mathchoice%
        {\vcenter{\hbox{\huge $\ovee$}}}%
        {\vcenter{\hbox{\Large $\ovee$}}}%
        {\ovee}{\ovee}}\displaylimits}

\newcounter{counter}
\theoremstyle{definition}
\newtheorem{theorem}[counter]{Theorem}
\newtheorem{proposition}[counter]{Proposition}
\newtheorem{lemma}[counter]{Lemma}
\newtheorem{definition}[counter]{Definition}
\newtheorem{corollary}[counter]{Corollary}
\newtheorem{example}[counter]{Example}
\newtheorem{remark}[counter]{Remark}

\usepackage[color,leftbars]{changebar}

\title{The three types of normal sequential effect algebras}
\author{Abraham Westerbaan}
\email{bram@westerbaan.name}
\affiliation{Radboud Universiteit Nijmegen}

\author{Bas Westerbaan}
\email{bas@westerbaan.name}
\affiliation{Radboud Universiteit Nijmegen}
\affiliation{University College London}

\author{John van de Wetering}
\email{john@vandewetering.name}
\affiliation{Radboud Universiteit Nijmegen}
\orcid{0000-0002-5405-8959}

\begin{document}

\maketitle

\begin{abstract}
	A sequential effect algebra (SEA) is an effect algebra equipped with a \emph{sequential product} operation modeled after the L\"uders product $(a,b)\mapsto \sqrt{a}b\sqrt{a}$ on C$^*$-algebras.
    A SEA is called \emph{normal}
    when
    it has all suprema of directed sets,
    and the sequential product interacts suitably with these suprema.
The effects on a Hilbert space
    and the unit interval of a von Neumann or JBW algebra
    are examples of normal SEAs
    that are in addition \emph{convex}, i.e.~possess a
    suitable action of the real unit interval on the algebra.
    Complete Boolean algebras form normal SEAs too,
    which are convex only when~$0=1$.

We show that any normal SEA~$E$
    splits as a direct sum $E= E_b\oplus E_c \oplus E_{ac}$
    of a complete Boolean algebra~$E_b$,
    a convex normal SEA~$E_c$,
    and a newly identified type of normal SEA~$E_{ac}$
    we dub \emph{purely almost-convex}.

Along the way we show, among other things,
that a SEA which contains only idempotents
must be a Boolean algebra;
and we establish a spectral theorem
using which we settle for the class of normal SEAs
a problem of Gudder regarding the uniqueness of square roots.
After establishing our main result,
    we propose a simple extra axiom for normal SEAs that
    excludes the seemingly pathological
    a-convex SEAs.
We conclude the paper by a study of
    SEAs with an associative sequential product.
We find that associativity forces normal SEAs satisfying our new
    axiom to be commutative, shedding light on the question of why the
    sequential product in quantum theory should be non-associative.
\end{abstract}

\section{Introduction}

Understanding the properties and foundations of quantum theory requires contrasting it with hypothetical alternative physical theories and mathematical abstractions. By studying these alternatives it becomes clearer which parts of quantum theory are special to it, and which are present in any reasonable physical theory.

A framework that has been used extensively to study such alternatives is that of \emph{generalised probabilistic theories} (GPTs)~\cite{barrett2007information}. GPTs have built into their definition the classical concepts of probability theory, and so the state and effect spaces of hypothetical physical systems are modelled by \emph{convex} sets.
While convexity is a useful and powerful property, it precludes the study of physical theories that have a more exotic notion of probability. For instance, deterministic or possibilistic `probabilities' studied in contextuality~\cite{abramsky2011sheaf,santos2020possibilistic} or systems where probabilities can vary spatially~\cite{moliner2017space}.

To study systems with this broader notion of probability, a more general structure than convex sets is needed.
This paper is about \emph{effect algebras}, which were introduced in 1994 by Foulis and Bennett~\cite{foulis1994effect}.

Effect algebras generalise and abstract the unit interval of effects in a C$^*$-algebra.
The study of effect algebras has become a flourishing field on its own~\cite{gudder1996examples,gudder1996effect,chajda2009every,ravindran1997structure,jenvca2001blocks,dvurevcenskij2010every,foulis2010type} 
    and covers a variety of topics~\cite{staton2018effect,roumen2016cohomology,jacobs2015effect,Jencova2019propertiesof}.
An effect algebra that has an action of the real unit interval on it is called a \emph{convex} effect algebra~\cite{gudder1999convex}. Such effect algebras essentially recover the standard notion of an effect space in GPTs.

Effect algebras, or even just convex effect algebras are very general structures. So, as in GPTs where additional assumptions are often needed on the physical theory to make interesting statements, we will impose some additional structure on effect algebras.

The effects of a quantum system, i.e.~the positive subunital operators on some Hilbert space or, more generally, positive subunital elements of a C$^*$-algebra or Jordan algebra, have some interesting algebraic structure
(aside from the effect algebra operations). The one we focus on in this paper is the \emph{sequential product}, also known as the \emph{L\"uders product} $(a,b)\mapsto \sqrt{a}b\sqrt{a}$. This product models the act of first measuring the effect $a$ and then measuring the effect $b$.

In order to study the sequential product in the abstract, Gudder and Greechie introduced in 2002~\cite{gudder2002sequential} the notion of a \emph{sequential} effect algebra (SEA), that models the sequential product on an effect algebra as a binary operation satisfying certain properties.
SEAs have been studied by several authors,
    see e.g.~\cite{gudder2005open,gudder2005uniqueness,gudder2008characterization,jun2009remarks,wetering2018characterisation,wetering2018sequential,shen2009not,giski2015entropy,giski2017introduction,jun2009sequential,tkadlec2008atomic,habil2008tensor,basmaster,weihua2009uniqueness}.

Adopting the language of GPTs (and specifically that of reconstructions of quantum theory) we can view the assumed existence of a sequential product on our effect algebras as a physical \emph{postulate} (indeed, such a sequential product has been used for a reconstruction of quantum theory before~\cite{wetering2018sequential}). 
But as is the case for GPTs, we will also need some structural `background' assumption. 
Indeed, in most work dealing with GPTs it is assumed that the sets of states and effects are closed in a suitable topology. 
This models the operational assumption that when we can arbitrarily closely approximate an effect, that this effect is indeed a physical effect itself. Such assumptions are routinely used to, for instance, remove the possibility of infinitesimal effects.

In a (sequential) effect algebra we however have no natural notion of topology and hence no direct way to require such a closure property.
But as an effect algebra is a partially ordered set, we can require that this order is \emph{directed complete}, meaning that any upwards directed subset has a supremum. 
Indeed, this condition is routinely used in theoretical computer science
    to give meaning to algorithms with loops or recursion
    in a branch called \emph{domain theory}~\cite{abramsky1994domain}.
In the field of operator algebras, one possible characterisation
    of von Neumann algebra's is that they are
    C$^*$-algebras that are bounded directed-complete
    and have a separating set of normal states~\cite{kadison1956operator,bramthesis,rennela2014towards}.

Motivated by the structure of von Neumann algebras, we call a sequential effect algebra \emph{normal} when every directed set has a supremum,
    the sequential product preserves these suprema in the second argument,
        and an effect commutes with such a suprema provided
        it commutes with all elements in the directed set.
The set of effects on a Hilbert space is a convex normal SEA.
More generally, the unit interval
    of any JBW-algebra (and so in particular that of any von Neumann algebra)
    is also a convex normal SEA~\cite{wetering2019commutativity}.
As a rather different example,
    any complete Boolean algebra is a normal SEA,
    which is not convex.

A priori, a (normal) SEA is a rather abstract object and there could potentially be rather exotic examples of them. The main result of this paper is to show that for normal SEAs this is, in a sense, \emph{not} the case.

We will show that any normal SEA
    is isomorphic to a direct sum~$E_b\oplus E_c \oplus E_{ac}$,
    where~$E_b$ is a complete Boolean algebra,
    $E_c$ is a convex normal SEA
    and~$E_{ac}$ is a normal SEA that is \emph{purely a-convex}, a new type of effect algebra we will define later on.
We will show there is no overlap:
    for instance, there is no normal SEA that is both purely a-convex and Boolean.

In much the same way as a C$^*$-algebra is the union of its commutative subspaces, we show that a purely a-convex SEA is a union of convex SEAs. Hence, normal SEAs come essentially in two main types: Boolean algebras that model classical deterministic logic, and convex SEAs that fit into the standard GPT framework.
Hence, rather than having convex structure as a starting assumption, we show that it can be derived as a consequence of our other assumptions.

Along the way we will establish several smaller results that might be of independent interest. 
We show any SEA where all elements are sharp (i.e.~idempotent) is a Boolean algebra, and that consequently any SEA with a finite number of elements is a Boolean algebra. 
We introduce a new axiom for SEAs that exclude the seemingly pathological purely a-convex SEAs. Combining this with our main result shows that normal SEAs satisfying this additional axiom neatly split up into a Boolean algebra and a convex normal SEA. 
Finally, we study SEAs where the sequential product is associative. 
We find that associative normal SEAs satisfying our additional axiom must be commutative, and we are able to completely classify the associative normal purely a-convex \emph{factors}, i.e.~SEAs with trivial center.  

This work relies on recent advances made in the representation theory
    of directed-complete \emph{effect monoids}~\cite{first}.
An effect monoid is an effect algebra with an additional
    associative (not necessarily commutative) multiplication
    operation that is additive in both arguments.
Crucially, any commutative normal SEA is a directed-complete
    effect monoid.

Section~\ref{sec:prelims} contains the basic definitions and recalls the necessary results from~\cite{first}.
Then in Section~\ref{sec:boolean} we show that a (normal) SEA where every element is idempotent must be a (complete) Boolean algebra.
In Section~\ref{sec:almostconvex} we prove our main technical results that show that any normal SEA splits up into a Boolean algebra and an a-convex normal SEA.
Then in Section~\ref{sec:purea-convex} we improve this result by showing that an a-convex normal SEA splits up into a convex part and a purely a-convex part, and we introduce a new axiom that excludes the seemingly pathological purely a-convex normal SEAs.
In Section~\ref{sec:assoc} we study the consequences of our representation theorem for the existence of non-commutative associative sequential products. 
Finally, in Section~\ref{sec:conclusion} we speculate on possible future avenues and consequences of our results.

\section{Preliminaries}\label{sec:prelims}

\begin{definition}
    An \Define{effect algebra}
    (EA)~\cite{foulis1994effect} is a set~$E$ with
    distinguished element~$0 \in E$,
    partial binary operation~$\ovee$ (called \Define{sum})
    and (total) unary operation~$a \mapsto~a^\perp$ (called \Define{complement}),
    satisfying the following axioms,
    writing~$a\perp b$ whenever~$a \ovee b$ is defined
        and defining~$1 := 0^\perp$.
\begin{itemize}
\item Commutativity: if $a\perp b$, then $b \perp a$ and $a\ovee b = b \ovee a$.
\item Zero: $a\perp 0$ and $a\ovee 0 = a$.
\item Associativity: if $a\perp b$ and $(a\ovee b)\perp c$, then
    $b\perp c$,
$a\perp (b \ovee c)$, and $(a\ovee b) \ovee c = a\ovee (b\ovee c)$.
\item The complement
    $a^\perp$ is the unique element with $a \ovee a^\perp = 1$.
\item If $a\perp 1$, then $a=0$.
\end{itemize}
For~$a,b \in E$ we write~$a \leq b$ whenever there is a~$c \in E$
    with~$a \ovee c = b$.
This turns~$E$ into a poset with minimum~$0$ and maximum~$1$.
The map~$a \mapsto a^\perp$ is an order anti-isomorphism.
Furthermore, $a \perp b$ if and only if~$a \leq b^\perp$.
If~$a \leq b$, then the element~$c$ with~$a \ovee c = b$
    is unique and is denoted by~$b \ominus a$.
\end{definition}

\begin{remark}
    We pronounce $a\perp b$ as `$a$ is summable with $b$'.
    In the literature this is more often referred to as orthogonality,
    but we avoid this terminology,
    as it clashes with the orthogonality 
    for sequential effect algebras
    we will see in Definition~\ref{defn:sea}.
    We use the symbol $\ovee$ to denote the effect algebra sum
    following for instance~\cite{jacobs2012coreflections,kentathesis}. The symbol~$\oplus$ is also commonly used in the literature.
\end{remark}

\begin{example}\label{ex:orthomodularlattice}
    Let~$B$ be a Boolean algebra (or more generally, an orthomodular
        poset~\cite{holland1975current,maeda,loomis}.)
    Then $B$ is an effect algebra with the partial addition defined by 
    $x\perp y \iff x\wedge y = 0$ and in that case  $x\ovee y = x\vee y$.
    The complement, $(\ )^\perp$, is given by the complement, $(\ )^\perp$
    (or the orthocomplement for an orthomodular lattice.)
    The lattice order coincides with the effect algebra order (defined
    above). See e.g.~\cite[Prop.~27]{basmaster}
    or~\cite[\S5]{foulis1994effect}.
\end{example}

\begin{example}\label{ex:cstaralgebra}
Let~$G$ be an ordered abelian group
    (such as an ordered vector space, for example a C$^*$-algebra,)
    and let $u\in G$ be any positive element.
    Then any interval~$[0,u]_G =
    \{a\in G~;~ 0\leq a \leq u\}$ forms an effect algebra, 
    where the effect algebra sum $a\ovee b$ of $a$ and $b$ 
    from the interval is defined when $a+b \leq u$ and 
    in that case coincides
    with the group addition: $a\ovee b = a+b$.
    The complement is
    given by $a^\perp = u-a$.
    The effect algebra order on~$[0,u]_G$ coincides
    with the regular
    order on~$G$.  
    
    In particular,
    the set of effects $[0,1]_C$
    of a unital C$^*$-algebra~$C$ forms an effect algebra
    with $a\perp b\iff a+b\leq 1$,
    and~$a^\perp = 1-a$.
\end{example}

\begin{remark}
In GPTs, the set of effects of a system is often modelled by the unit interval of an ordered vector space. Hence, the above example demonstrates how these effect spaces fit into the notion of an effect algebra. The corresponding addition of effects in a GPT is often called \emph{coarse-graining}.
\end{remark}

\begin{example}\label{eaprod}
If~$E$ and~$F$ are effect algebras, then
    its Cartesian product~$E \oplus F$
    with component-wise operations
    is again an effect algebra.\footnote{This is in fact 
    a categorical product
        with the obvious projectors and as morphisms
        maps~$f$ with: $f(1)=1$ and~$a \perp b$ implies~$f(a) \perp f(b)$
        and~$f(a)\ovee f(b) = f(a\ovee b)$.}
\end{example}

We will often implicitly use the following basic facts about effect algebras.
\begin{proposition}
    In any effect algebra we have (see e.g.~\cite{dvurecenskij2013new} or
        \cite[\S175V]{basthesis})
    \begin{enumerate}
        \item \emph{(involution)} $a^{\perp\perp} = a$;
        \item \emph{(positivity)} if $a\ovee b = 0$, then~$a=0$ and~$b=0$;
        \item \emph{(cancellation)} if~$a \ovee b = a \ovee c$, then~$b = c$;
        \item $a \leq b$ iff $b^\perp \leq a^\perp$ \emph{and}
        \item $a \perp b$ iff $a \leq b^\perp$.
    \end{enumerate}
\end{proposition}

Some effect algebras are more closely related to ordered vector spaces (like the example of a C$^*$-algebra above or the effect spaces coming from a GPT). We call these effect algebras convex:
\begin{definition}\label{def:convex}
  A \Define{convex action}
  on an effect algebra~$E$ is 
   a map $\cdot: [0,1]\times E\rightarrow E$, 
   where $[0,1]$ is the regular unit interval, 
   obeying
    the following axioms for all $a,b\in E$ and~$\lambda, \mu \in
    [0,1]$:
    \begin{itemize}
        \item $\lambda\cdot (\mu\cdot a) = (\lambda\mu)\cdot a$.
    \item If $\lambda+\mu \leq 1$, then $\lambda \cdot a \perp
    \mu \cdot a$ and $\lambda \cdot a \ovee \mu \cdot a =
    (\lambda+\mu)\cdot a$.
        \item $1\cdot a = a$.
        
        \item $\lambda\cdot (a\ovee b) \ = \ \lambda\cdot a \ovee \lambda\cdot b$.
    \end{itemize}
A \Define{convex effect algebra} \cite{gudder1999convex}
is an effect algebra endowed with such a convex action\footnote{Usually, a convex effect algebra is defined simply as an effect algebra $E$ where for every $a\in E$ and $\lambda\in[0,1]$ there is an element $\lambda a\in E$ satisfying the axioms above. However, in this paper we face the possibility of an effect algebra having many inequivalent convex actions. This is why we are explicit in defining a convex effect algebra as an effect algebra equipped with a convex action.}.
  We will say that an effect algebra~$E$ is \Define{convex}
  when there is at least one convex action on~$E$.
\end{definition}

\begin{example}
    Let $V$ be an ordered real vector space (such as the space of self-adjoint elements of a C$^*$-algebra). Then any interval $[0,u]_V$ where $u\geq 0$ is a convex effect algebra with the obvious action of the real unit interval. 
    Conversely, for any convex effect algebra $E$, we can find an ordered real vector space $V$ and $u\in V$ such that $E$ is isomorphic as a convex effect algebra to $[0,u]_V$~\cite{gudder1998representation}.
\end{example}

\begin{remark}
    Convex effect algebras have been well-studied, see e.g.~\cite{jacobs2011probabilities,gudder2018convex,wetering2018characterisation,wetering2018reconstruction,Jencova2019propertiesof}. 
  In the literature on effectus theory, 
  convex effect algebras are also 
  called \emph{effect modules}~\cite{cho2015introduction,kentathesis}.
In Definition~\ref{def:almostconvex} 
  we introduce \emph{a-convex} (almost convex) 
  effect algebras, by dropping the last axiom.
\end{remark}

\begin{definition}
    Let $E$ be a partially ordered set (such as an effect algebra).
    A subset~$S$ of~$E$ is called \Define{directed}
    when it is non-empty,
    and for any  $a,b\in S$ there exists a $c\in
    S$ such that $a,b\leq c$. We say that~$E$ is \Define{directed complete} when
    every directed subset~$S$ of~$E$ has a supremum, $\bigvee S$.
\end{definition}

\begin{example}
Any complete Boolean algebra is a directed-complete
effect algebra.
\end{example}

\begin{example}
    Let~$\mathscr{A}$ be a unital C$^*$-algebra.
Then $[0,1]_{\mathscr{A}}$ is a directed-complete effect algebra
    if and only if~$\mathscr{A}$
    itself  is bounded-directed complete, that is: if every
    bounded set of self-adjoint elements has a least upper bound.
Such C$^*$-algebras are called \emph{monotone complete}
    or \emph{monotone closed}, \cite[\S2]{kadison1971equivalence},
    and include all von Neumann algebras.
A commutative unital C$^*$-algebra,
    being of the form $C(X)$ for some compact Hausdorff space~$X$,
    is bounded-directed complete
    if and only if~$X$ is extremally disconnected~\cite[Sections~1H \& 3N.6]{gillman2013rings}.
\end{example}

\begin{remark}\label{remark:infima}
	Because the complement acts as an order anti-isomorphism,
	a directed-complete effect algebra also has infima for all
	\Define{filtered}, i.e.~downwards directed, sets. In
	particular, if we have a decreasing sequence $a_1\geq a_2\geq
	a_3\geq \cdots$ in a directed-complete effect algebra, then
	this has an infimum $\bigwedge_n a_n$.
\end{remark}

\subsection{Sequential effect algebras}

We now come to the definition of the central object of study of this paper, 
the sequential effect algebra. Before we give the formal definition, let us give some motivation for the assumptions we will require.
The sequential product $a\mult b$ of two effects $a$ and $b$ represents the sequential measurement of first $a$ and then $b$.
An important difference between classical and quantum systems is that in a classical system we can measure without disturbance, and hence the order of measurement is not important: $a\mult b = b\mult a$ for all effects $a$ and $b$.
In a quantum system this is generally not the case, and the order of measurement is important.
However, what is interesting in quantum theory is that some measurements are \emph{compatible}, meaning that the order of measurement for those measurements is not important.%
\footnote{The word `compatible' is used for this purpose when it comes to sequential effect algebras~\cite{gudder2002sequential}. However, in other parts of the literature `compatible' is used for a different notion, and the term `non-disturbing' might be used instead~\cite{Heinosaari2019nofreeinformation}.} 
We will use the symbol $a\commu b$ to denote that $a$ and $b$ are compatible effects. 
For such compatible effects $a$, $b$ and $c$ we can expect that doing `classical operations' or `classical post-processing' on these effects retains compatibility. For instance, if $a\commu b$, and we negate the outcome of $a$ to get $a^\perp$, we still expect $a^\perp \commu b$, as negation can be done classically without interacting with the quantum system. Similarly, we would expect $a\commu b\ovee c$, as addition of effects corresponds to classically coarse-graining the measurement outcomes.

Let us now give the formal definition:
\begin{definition}
  \label{defn:sea}
    A \Define{sequential effect algebra}
    (SEA)~\cite{gudder2002sequential}~$E$
    is an
     effect algebra with an additional (total)
    binary operation~$\mult$,
        called the \Define{sequential product},
        satisfying the axioms listed below,
        where $a,b,c\in E$.
        Elements~$a$ and~$b$ are said to \Define{commute},
            written~$a \commu b$,
            whenever~$a \mult b = b \mult a$.
    \begin{enumerate}
        \item[S1.]
            $a\mult (b\ovee c) = a\mult b \ovee a \mult c$
            whenever~$b\perp c$.
        \item[S2.]
            $1\mult a = a$.
        \item[S3.] 
            $a\mult b = 0 \implies b\mult a =0$.
        \item[S4.]
            If $a\commu b$, then $a\commu b^\perp$ and $a\mult
                (b\mult c) = (a\mult b)\mult c$ for all $c$.
        \item[S5.] 
                If $c\commu a$ and $c\commu b$ then also $c\commu
	    a\mult b$ and if furthermore $a\perp  b$, then $c\commu a\ovee
	    b$.
    \end{enumerate}
    A SEA~$E$ is called \Define{normal}
        when~$E$ is directed complete, and
    \begin{enumerate}[resume]
        \item[S6.]
            Given directed~$S\subseteq E$ we have
            $a\mult \bigvee S = \bigvee_{s\in S} a\mult s$,
            and $a\commu \bigvee S$ when $a\commu s$ for all~$s\in S$.
    \end{enumerate}
    We say~$E$ is \Define{commutative} whenever~$a \commu b$ for
        all~$a,b \in E$.
    An element~$a \in E$ is \Define{central}
        if it commutes with every element in~$E$.
        The \Define{center}~$Z(E)$ of~$E$ is the set of all central elements.
    We call~$p \in E$ an \Define{idempotent}
        whenever~$p^2=  p\mult p = p$
    (or equivalently~$p \mult p^\perp = 0$).
    We call~$E$ \Define{Boolean}
        if every element is an idempotent\footnote{
        We will see that a Boolean SEA is a Boolean algebra,
         see Proposition~\ref{prop:SEAsharpisBoolean}.}.
    We say~$a,b \in E$ are \Define{orthogonal},
        provided that~$a \mult b =0$.
\end{definition}

  The discussion above Definition~\ref{defn:sea}
  should make clear why we have axioms S2, S4 and S5. 
Axiom S3 is less operationally motivated, but is a useful property as it makes the notion of orthogonality well-behaved. 
Axiom S6 is our way of stating continuity of the product without having access to an actual topology. 
Indeed, if $E$ is the unit interval of a von Neumann algebra, then Axiom S6 says  that the product is ultraweakly continuous in the second argument.
Axiom S1 can be understood operationally as follows.
Consider an ensemble of systems all prepared in the same state and that we measure them all with the effect $a$. We then measure some of the states with $b$ and others with $c$. By adding up the probabilities of success of these measurement outcomes we model the effect $a\mult (b\ovee c)$ as we have `$a$ and then $b$ or $c$'. Alternatively, we could measure one half of the states with $a\mult b$ and the other with $a\mult c$ and add up the outcomes to get $(a\mult b) \ovee (a\mult c)$. But looking at the interactions we have with the system, these two protocols are obviously equivalent, and hence we should have $a\mult (b\ovee c) = (a\mult b) \ovee (a\mult c)$.
Note that we do not in general expect this property to hold in the other argument: $(a\ovee b)\mult c \neq (a\mult c) \ovee (b\mult c)$. Indeed, the addition is a classical operation, and there is no way in general to perform the measurement $(a\ovee b)\mult c$ as it requires the measurement `$a$ or $b$' to be performed on a single system.

\noindent The motivating example of a (normal) sequential effect algebra is the unit interval of a C$^*$-algebra:

\begin{example}
    \label{ex:canonical-sea}
If $\mathscr{A}$ is a C$^*$-algebra, then the EA~$[0,1]_{\mathscr{A}}$ is a SEA with
    the sequential product defined by~$a \mult b := \sqrt{a} b \sqrt{a}$,
    see~\cite{gudder2002sequential}.
    If $\mathscr{A}$ is furthermore bounded-directed complete
        (for instance if it is a von Neumann algebra),
    then $[0,1]_\mathscr{A}$ is a normal SEA
    with the same sequential product.
    It is also possible to define a sequential product using only the Jordan algebra structure. In particular, any JB-algebra is a convex SEA, while any JBW-algebra is a convex normal SEA. For the details we refer to~\cite{wetering2019commutativity}.
\end{example}

We will use the following properties without further reference in the remainder of the paper.
\begin{proposition}[Cf.~\S3 of \cite{gudder2002sequential}]\label{prop:SEAbasicproperties}
    Let $E$ be a SEA with $a,b,c,p \in E$
        with~$p$ idempotent.
    \begin{enumerate}
        \item
            \label{prop:SEAbasicproperties-1}
            $a\mult 0 = 0 \mult a = 0$ and $a\mult 1 = 1\mult a = a$.
        \item 
            \label{prop:SEAbasicproperties-2}
            $a\mult b \leq a$.
        \item 
            \label{prop:SEAbasicproperties-3}
            If $a\leq b$, then $c\mult a\leq c\mult b$.
        \item
            \label{prop:SEAbasicproperties-4}
            $p\leq a$ iff $p\mult a =p$ iff $a\mult p = p$
            iff $a^\perp\circ p = 0$
            iff $p \circ a^\perp = 0$.
        \item
            \label{prop:SEAbasicproperties-5}
            $a\leq p$ iff $p\mult a=a$ iff $ a\mult p = a$
            iff $a\circ p^\perp = 0$ iff $p^\perp\circ a=0$.
        \item
            \label{prop:SEAbasicproperties-6}
            $p^\perp$ is idempotent.
        \item
            \label{prop:SEAbasicproperties-7}
            If $p \perp a$, then~$a$ is idempotent if and only if~$p \ovee a$
                is idempotent.
    \end{enumerate}
\end{proposition}
\begin{proof}
Concerning~\ref{prop:SEAbasicproperties-1}:
Since~$0=0\ovee 0$ by~S1,
we have $a\mult 0 = a\mult 0\,\ovee\, a\mult 0$,
and so~$a\mult 0 = 0$, by cancellativity of the addition in an effect algebra.
Then~$0\mult a=0$ too, by~S3.
In particular, $a\commu 0$, and so~$a\commu 0^\perp = 1$, by~S4.
With~S2, we get $a\mult 1 = 1\mult a = a$.

Point~\ref{prop:SEAbasicproperties-3}:
when~$a\leq b$, we have $b=a\ovee (b\ominus a)$,
    and so~$c\mult b = c\mult a\,\ovee\, c\mult(b\ominus a)
    \geq c\mult a$,
    by~S1. 
    Taking~$b=1$,
    we get~$c\mult a \leq c\mult 1 = c$,
    by~\ref{prop:SEAbasicproperties-1},
    and so we have~\ref{prop:SEAbasicproperties-2}.

To show pt.~\ref{prop:SEAbasicproperties-5},
    we will prove that we have the implications~$
    a \leq p \Rightarrow
    p^\perp \mult a = 0 \Rightarrow
    a \mult p^\perp = 0 \Rightarrow
    a \mult p = a \Rightarrow
    p \mult a = 0 \Rightarrow
    a \leq p$ and thus they are all equivalent.
    So first suppose that~$a\leq p$.
    Then~$p^\perp \mult a \leq p^\perp \mult p = 0$, 
    by~\ref{prop:SEAbasicproperties-3},
    so~$p^\perp \mult a = 0$.
Now suppose instead that $p^\perp \mult a=0$, which
is  equivalent to $a \mult p^\perp=0$ by~S3,
    which in turn is equivalent to~$a=a\mult 1 = a\mult (p\ovee p^\perp) = a\mult p$ by~S2.
Since then~$a\commu p^\perp$,
we get $a\commu p$, by~S4,
and so~$p \mult a = a\mult p = a$.
Since~$p\mult a=a$ on its own entails that~$a\leq p$,
by~\ref{prop:SEAbasicproperties-2},
we are back where we started, and therefore done.

For~\ref{prop:SEAbasicproperties-4},
note that~$p\leq a$ iff
$a^\perp \leq p^\perp$
iff $a^\perp \circ p = 0$
iff $p\circ a^\perp = 0$
by~\ref{prop:SEAbasicproperties-5}.
Moreover, since~$p\mult a=p$ is equivalent to~$p\mult a^\perp=0$ by~S1,
    and $a\mult p = p$ entails $p\leq a$ by~\ref{prop:SEAbasicproperties-2},
the only thing left to show is that $p\leq a$ implies 
    $a\mult p = p$.
So suppose that~$p\leq a$.
We already know that~$a^\perp \mult p = 0 = p \mult a^\perp = 0$
and~$p \mult a = p$.  
Since thus~$a\commu p^\perp$,
and so~$a\commu p$ by~S4,
we get~$a\mult p = p\mult a= p$.
We continue with pt.~\ref{prop:SEAbasicproperties-6}.
Note that~$a$ is idempotent iff~$a \mult a^\perp = 0$.
Thus~$p \mult p^\perp = 0$
    and so~$p^\perp \mult p = 0$ by S3.
    Hence~$p^\perp$ is indeed idempotent.

Finally, we move to point \ref{prop:SEAbasicproperties-7}.
Assume~$a$ is idempotent.
Clearly~$p \leq p \ovee a$
    and so~$(p \ovee a) \mult p = p$ by~\ref{prop:SEAbasicproperties-4}.
    Similarly~$(p \ovee a) \mult a = a$.
    Thus~$(p \ovee a) \mult (p \ovee a) = ((p \ovee a) \mult p)
    \ovee ((p \ovee a) \mult a) = p \mult a$, as desired.
Conversely, assume~$p \ovee a$ is an idempotent.
Note~$p^\perp$ and~$(p \ovee a)^\perp$ are summable idempotents and
    so by the previous point~$a^\perp = p^\perp \ovee (p \ovee a)^\perp$
    is idempotent as well. Thus~$a$ is indeed idempotent.
\end{proof}

The next five lemmas were originally
    proven in~\cite{first} for effect monoids,
    and we will need them for our results in the context of SEAs.

\begin{lemma}\label{lem:summableunderidempotent}
Let $E$ be a SEA with $p,a,b \in E$ and $p$ idempotent.
    If~$a,b\leq p$ and~$a\ovee b$ exists,
then~$a\ovee b\leq p$.
\end{lemma}
\begin{proof}
Because $a\leq p$, we have $p^\perp \mult a = 0$.
Since similarly, $p^\perp\mult b=0$,
    we have $p^\perp\mult (a\ovee b)
    = 0$, and so $a\ovee b \leq p$.
\end{proof}

\begin{lemma}\label{lem:selfsummable}
    In a SEA, $a \mult a^\perp$ is summable with itself for any element~$a$.
\end{lemma}
\begin{proof}
Note that since~$a\commu a$,
we have~$a\commu a^\perp$ by S4,
or in other words  $a\mult a^\perp = a^\perp \mult a$.
Since
    $1=a\ovee a^\perp 
    = a\mult (a\ovee a^\perp) \,\ovee\, a^\perp\mult (a\ovee a^\perp)
    = a\mult a\,\ovee\, a\mult a^\perp\,\ovee\, a^\perp \mult a
    \,\ovee\, a^\perp \mult a^\perp =
    a^2 \ovee 2(a \mult a^\perp )\ovee (a^\perp)^2$,
    we see that~$a\mult a^\perp$ is indeed summable with itself.
\end{proof}

\begin{lemma}\label{lem:addition-normal}
    Let $E$ be a directed-complete effect algebra, and let $S\subseteq E$ be some directed non-empty subset. Let $a\in E$ be such that $a\perp s$ for all $s\in S$. Then $a\perp \bigvee S$ and $a\ovee \bigvee S = \bigvee_{s\in S} a\ovee s$.
\end{lemma}
\begin{proof}
    The map $b\mapsto a\ovee b$, giving an order 
    isomorphism from $[0, a^\perp]_E$ to $[a, 1]_E$ 
    (with inverse $b\mapsto b\ominus a$,)
    preserves (and reflects) suprema.
    In particular, $\bigvee S$,
    the supremum of~$S$ in~$E$,
    which is the supremum of~$S$ in~$[0,a^\perp]_E$ too,
    is mapped to $\bigvee_{s\in S} a\ovee s$,
    the supremum of the $a\ovee s$ in $E$, and in~$[a,1]_E$ too.
\end{proof}

\begin{lemma}
\label{lem:archemedeanomegadirectedcomplete}
The only element~$a$ of a directed-complete
effect algebra~$E$
for which the $n$-fold sum~$na$ exists for all~$n$
is zero.
\end{lemma}
\begin{proof}
We have $a\ovee \bigvee_n na
    =\bigvee_n a\ovee na
    = \bigvee_n (n+1)a = \bigvee_n na$,
    and so~$a=0$.
\end{proof}

\begin{lemma}\label{lem:nonilpotents}
    Let $E$ be a normal SEA and suppose $a\in E$
    satisfies~$a^2 = 0$.
        Then~$a = 0$.
\end{lemma}
\begin{proof}
    Since $a^2 = 0$ we have $a = a\mult 1 = a\mult (a\ovee a^\perp) = a\mult a^\perp$, and hence (see
    Lemma~\ref{lem:selfsummable}) $a$ is summable with itself. But
    furthermore~$(a\ovee a)^2 = 4a^2 = 0$, and so $(a\ovee a)^2=0$.

    Continuing in this fashion,
        we see that~$2^n a$ exists for every~$n \in \N$
        and~$(2^n a)^2 = 0$.
        Hence, for any~$m \in \N$
        the sum~$m a$ exists 
    so that by Lemma~\ref{lem:archemedeanomegadirectedcomplete},
    we have $a = 0$.
\end{proof}

\begin{proposition}
Assume~$E$ is a SEA with idempotent~$p \in E$.
Write~$p \mult E := \{ p\mult a; \ a \in E\}$
    for the \Define{left corner} by~$p$.
The sequential product of~$E$ restricts to~$p \mult E$
    and in fact, with partial sum and zero of~$E$
    and complement~$a \mapsto p \ominus a$,
    the set~$p \mult E$ is an SEA.
If~$E$ is normal, then~$p \mult E$ is normal as well.
\end{proposition}
\begin{proof}
By Prop.~\ref{prop:SEAbasicproperties}
    we have~$p \mult E = \{ a; \ a\in E\; a \leq p \}$
    and so~$p \mult E$  is an effect algebra.
For~$a,b \in p \mult E$,
    we have~$a \mult b \leq a \leq p$ and so~$a \mult b \in p \mult E$.
As the sequential product, zero and addition of~$p \mult E$ and~$E$ coincide
    almost all axioms for an SEA hold trivially.
    Only the first part of~S4 (which involves the orthocomplement) remains.
So assume~$a,b \in p\mult E $ with~$a \commu b$.
We have to show that~$a \commu p \ominus b$.
Note~$a \mult p^\perp = 0$ and so~$a \commu p^\perp$.
    Thus by S5, we have~$a \commu b \ovee p^\perp$ and as $b\ovee p^\perp = (p \ominus b)^\perp$ we have by~S4 (for~$E$),~$a \commu p \ominus b$.
        Thus~$p \mult E$ is indeed an SEA.

Now assume~$E$ is normal.
As a principal downset of~$E$,
    the suprema computed within~$p \mult E$
        are the same as computed in~$E$
        and so~$p \mult E$ is directed complete.
As additionally the sequential product of~$p \mult E$
    is the restriction of that of~$E$,
    the axiom S6 holds trivially.
\end{proof}

\begin{proposition}\label{prop:central-splits}
    Let $p$ be a central idempotent in a SEA $E$. Then $E\cong p\mult E \oplus p^\perp \mult E$.
\end{proposition}
\begin{proof}
    Note that the map $a\mapsto (p\mult a, p^\perp \mult a)$ is
    additive and unital. It is order reflecting, because if $p\mult
    a \leq p\mult b$ and $p^\perp \mult a \leq p^\perp \mult b$
    then $a=a\mult p \ovee a\mult p^\perp = p\mult a \ovee p^\perp
    \mult a \leq p\mult b \ovee p^\perp\mult b = b$, using the
    centrality of $p$ (and hence $p^\perp$ by S4). It is obviously
    surjective because for $a\in p\mult E$ we have $p\mult a = a$,
    and similarly for $b\in p^\perp \mult E$, and hence $a\ovee b
    \mapsto (a,b)$.

    To show that it preserves the sequential product we note that 
    \begin{equation*}
    p\mult (a\mult b) \ =\  (p\mult a)\mult b \ =\  (p\mult (p\mult
    a))\mult b  \ =\  ((p\mult a)\mult p)\mult b \ = \  (p\mult
    a)\mult (p\mult b)
    \end{equation*}
    and similarly for $p^\perp$.
\end{proof}

\begin{definition}
    We call an idempotent $p\in E$ \Define{Boolean} when $p\mult E$ is Boolean, i.e.~when all $a\leq p$ in $E$ are idempotent.
\end{definition}

The following definition and result will be crucial for our arguments, as it relates sequential effect algebras to commutative effect monoids (see next section).

\begin{definition}
    Let $S\subseteq E$ be a subset of elements of a SEA $E$.
        The \Define{commutant of $S$}
            is defined as $S^\prime := \{a\in E;\ 
                    a \commu s \text{ for all } s \in S \}$.
    The \Define{bicommutant of~$S$} is
        defined simply as~$S'' := (S')'$.
\end{definition}

\begin{remark}
    Let $\mathscr{H}$ be some Hilbert space, and
    let $E=[0,1]_{B(\mathscr{H})}$.
    For any $a\in E$, the bicommutant~$\{a\}''$ 
    is the set of effects of the least 
    commutative von Neumann algebra of~$B(\mathscr{H})$ containing~$a$
    by the bicommutant theorem.
In general this is false,
    consider for instance~$E=[0,1]^2$ (with component-wise standard product)
    and~$a = (0,1)$
    --- then~$\{a\}'' = E$, while~$\{0\}\times [0,1]$
    is a smaller commutative subalgebra containing $a$.
\end{remark}

\begin{proposition}[Cf.~\cite{wetering2018characterisation}, Proposition III.12]\label{prop:doublecommutant}
Let $S\subseteq E$ be a set of mutually commuting elements in
    a normal SEA $E$, then $S^{\prime\prime}$ is a commutative normal SEA.
\end{proposition}
\begin{proof}
    To start off, assume~$S \subseteq E$ an arbitrary subset of~$E$.
    Clearly~$0,1 \in S'$.
    If~$a,b\in S'$, then for any~$s \in S$,
        we have~$a\commu s$ and $b\commu s$ so that $a\mult b \commu s$
            and~$a^\perp \commu s$.
    If~$a \perp b$, then~$a \ovee b \commu s$.
    Thus~$S'$ is closed under partial sum, complement and sequential product.
    As a direct axiom of a normal SEA,
        the set~$S'$ is also closed under directed suprema.
    Hence~$S'$ is a sub-normal-SEA of~$E$.

    As~$S$ was arbitrary, we see that~$(S')' =  S''$ is a sub-normal-SEA of~$E$ as well.
    Assume the elements of~$S$ are pairwise commuting.
    Then~$S \subseteq S'$ so that~$S'' \subseteq S'$.
    By definition~$S''$ commutes with all elements from~$S'$
        and so in particular with all elements from~$S''$ itself.
        Hence~$S''$ is commutative.
\end{proof}


\subsection{Effect monoids}
Before we can proceed with the theory of (normal) SEAs
we must discuss effect monoids.
Roughly speaking, effect monoids are the classical counterparts
to SEAs, being endowed with an associative and biadditive 
multiplication $\cdot$.
Note that we do not require an effect monoid to be commutative,
although commutativity often follows automatically,
for example, in the (for us relevant) case that the effect monoid is
directed complete.
\begin{definition}\label{def:effectmonoid}
    An \Define{effect monoid} (EM) is an effect algebra $(M,\ovee, 0, (\ )^\perp,
    \,\cdot\,)$ with an additional (total) binary operation $\cdot$,
such that the following conditions hold 
    for all~$x,y,z \in M$.
\begin{itemize}
\item Unit: $x\cdot 1 = x = 1\cdot x$.
\item Distributivity: If $y\perp z$, then $x\cdot y\perp  x\cdot z$,\quad
    $y\cdot x\perp z\cdot x$
        with~$x\cdot (y\ovee z) = (x\cdot y) \ovee (x\cdot z)$
        and~$ (y\ovee z)\cdot x = (y\cdot x) \ovee (z\cdot x).$
In other words:
        $\,\cdot\,$ is bi-additive.
\item Associativity: $x \cdot (y\cdot z) = (x \cdot y) \cdot z$.
\end{itemize}
We call an effect monoid~$M$ \Define{commutative} if $x\cdot y = y \cdot x$
    for all~$x,y \in M$;
an element $p$ of~$M$  \Define{idempotent}
whenever~$p^2 = p \cdot p = p$;
elements~$x$, $y$  of~$M$ \Define{orthogonal}
when $x\cdot y = y\cdot x = 0$.
    An effect monoid is \Define{Boolean}
        if all its elements are idempotents.

\end{definition}

\begin{example}\label{ex:booleanalgebra}
Any Boolean algebra  $(B,0,1,\wedge,\vee,(\ )^\perp)$
    is an effect algebra by Example~\ref{ex:orthomodularlattice},
and, moreover,  a Boolean commutative effect monoid with
    multiplication defined by $x \cdot y= x \wedge y$.
    Conversely, any Boolean effect monoid is a Boolean algebra~\cite[Proposition~47]{first}.
\end{example}

\begin{example}\label{ex:CX}
    The unit interval $[0,1]_R$ of any (partially)
    ordered unital ring~$R$
    (in which the sum $a+b$ and product $a\cdot b$
    of positive elements~$a$ and~$b$ are again positive)
    is an effect monoid.

    For example, let $X$ be a compact Hausdorff space.
    The space of continuous complex-valued functions~$C(X)$
    is a commutative unital C$^*$-algebra
    (and conversely by the Gelfand theorem, any
    commutative C$^*$-algebra with unit is of this form) and hence its
    unit interval $[0,1]_{C(X)} = \{f:X\rightarrow [0,1]\}$ is a
    commutative effect monoid.

    In~\cite[Ex.~4.3.9]{kentathesis} and~\cite[Cor.~51]{basmaster}
        two different non-commutative effect monoids are constructed.
The latter also provides an example
of an effect monoid that is not a SEA
    (since $e_3 \cdot e_2=0$, while $e_2\cdot e_3=e_5$).
\end{example}

Note that an effect monoid that additionally satisfies the implication
$a\cdot b=0\implies b\cdot a = 0$
is a SEA.

\begin{example}\label{prop:commutativemonoidissequential}
	A commutative SEA is a commutative effect monoid and vice versa.
    Moreover, a normal commutative SEA is a directed-complete commutative effect monoid and vice versa\footnote{While the first claim is obvious, the astute reader might notice that the second claim requires showing that the product in a directed-complete commutative effect monoid is always normal. This is done in~\cite[Theorem~43]{first}.}.
\end{example}

\begin{remark}
	In Ref.~\cite{habil2008tensor}, ``distributive'' sequential effect algebras were introduced. These are the same thing as effect monoids satisfying the condition $a\cdot b = 0 \iff b\cdot a = 0$.
\end{remark}

\begin{example}\label{cornersexample}
    Let~$M$ be an effect monoid and let~$p \in M$ be some idempotent.
    The subset~$pM := \{p \cdot e; \ e\in M\}$
        is called the \Define{left corner} by~$p$
        and is an effect monoid with~$(p\cdot e)^\perp := p \cdot e^\perp$
            and all other operations inherited from~$M$.
        The map~$e \mapsto (p\cdot e, p^\perp \cdot e)$
        is an isomorphism~$M \cong p M \oplus
            p^\perp M$~\cite[Corollary~21]{first}.
Analogous facts hold for the right corner~$Mp := \{e \cdot p; \ e\in M\}$.
\end{example}

\subsubsection{Representation theorem for directed-complete effect monoids}\label{sec:representation-theorem}
We will need the following representation theorem
for directed-complete effect monoids from~\cite{first}.
(In fact,
a representation theory for the more general class
of $\omega$-directed complete
effect monoids is established there,
which we do not need here).
\begin{theorem}{\cite{first}}\label{thm:first}
In every directed-complete effect monoid~$M$
there is an idempotent~$p$
    such that~$p M$ is a convex effect algebra
        and~$p^\perp M$ is Boolean.
    Furthermore, there is an extremally-disconnected
        compact Hausdorff space~$X$
        such that~$pM \cong [0,1]_{C(X)}$.
\end{theorem}

\begin{corollary}[Spectral Theorem for normal SEA]\label{seaspectral}
    Let~$E$ be a normal SEA and let~$a \in E$.
    Then there is a extremally-disconnected compact Hausdorff space~$X$
        and complete Boolean algebra~$B$
        such that~$\{a\}'' \cong [0,1]_{C(X)} \oplus B$.
\end{corollary}
\begin{proof}
    Combine Example~\ref{prop:commutativemonoidissequential},
    Proposition~\ref{prop:doublecommutant}
    and Theorem~\ref{thm:first}.
\end{proof}

\noindent The previous allows us to answer Problem 20
    of~\cite{gudder2005open} for the
special case of normal SEAs.

\begin{corollary}
    Any element~$a$ of a normal SEA~$E$ has a unique square root.
\end{corollary}
\begin{proof}
    Note that~$a \in \{a\}''$.
    There is an idempotent~$p \in \{a\}''$
        such that~$p \mult \{a\}''$ contains only idempotents
        and~$p^\perp \mult \{a\}'' \cong [0,1]_{C(X)}$
            for some extremally-disconnected compact Hausdorff space~$X$.
    Write~$a_i := p \mult a$
        and~$a_c := p^\perp \mult a$.
    Then there is a unique~$b \in p^\perp \mult \{a\}''$ with~$b^2 = a_c$ (because this is true in $C(X)$).
    Define~$\sqrt{a} := a_i \ovee b$.
    As $a_i$ is idempotent, it is easy to see that~$\sqrt{a}^2 = a$
        and that in fact~$\sqrt{a}$ is the unique such element
        within~$\{a\}''$.

    To prove uniqueness, assume $c^2 = a$ for some $c\in E$. As $\{c\}''$ is a sub-algebra of mutually commuting elements we have $a=c^2 \in \{c\}''$ and hence $a\commu c$.
        As~$\sqrt{a} \in \{a\}''$,
        we must then also have~$c \commu \sqrt{a}$.
    Consider~$B := \{\sqrt{a},a, c\}''$.
    Reasoning as before, $a$ has a unique root in~$B$,
        hence~$c = \sqrt{a}$.
\end{proof}

Note that if a SEA is not normal, it does not need to have unique square roots~\cite{shen2010n}.

\subsection{An interesting sequential effect algebra}

Before we continue, let us construct a concrete example of a sequential effect algebra that is not directed complete, in order to serve as a foil of some of the other properties we will prove later. First, the following is a construction for effect monoids.
\begin{example}\label{ex:effect-monoids}
    Let $V$ be an ordered vector space. Let $R$ be the space of linear functions $f:V\rightarrow V$. For $f,g\in R$ we set $f\leq g$ when $f(v)\leq g(v)$ for all $v\geq 0$ in~$V$. This makes $R$ into an ordered vector space. 
    The set $M := [0,\id]_R := \{f\in R~;~ 0\leq f\leq \id\}$ equipped with the addition $f\ovee g := f+g$ that is defined when $f\leq \id-g$ is a convex effect algebra~\cite{gudder1999convex}. Furthermore, defining a product via  composition, $f\cdot g := f\circ g$, makes it an effect monoid.
    Indeed, this product obviously distributes over the addition and has $\id$ as the identity. That $0\leq f\cdot g \leq \id$ follows because for all $v\in V_+$
we have $0\leq f(g(v)) \leq f(v) \leq v = \id(v)$, since $0\leq g(v),f(v)\leq v = \id(v)$ by assumption.
\end{example}
We can use this construction to construct a specific effect monoid that also happens to be a non-commutative sequential effect algebra.
\begin{example}\label{ex:assoc-SEA}
    Consider~$V := \R^2$ with the positive cone defined 
    by $(a,b) > 0$ iff  $a+b > 0$.
    Let~$R$ and $M$ be defined as in Example~\ref{ex:effect-monoids}.
    Of course $R$ is just the space of $2\times 2$ real matrices.
    With some straightforward calculation it can then be verified
    that
\begin{equation*}
    A  =  \begin{pmatrix}
        a&b\\
        c&d
    \end{pmatrix} \in M \quad \iff\quad  A=0 \ \text{or}\ A=\id\  \text{or}\ 1>a+c=b+d>0.
\end{equation*}
    Define a map $\tau: M\rightarrow [0,1]$ by $\tau(\begin{psmallmatrix}a&b\\c&d\end{psmallmatrix}) = a+c=b+d$. Then it is straightforward to check that $\tau$ is monotone ($A\leq B \implies \tau(A)\leq \tau(B)$), multiplicative ($\tau(A\cdot B) = \tau(A)\tau(B)$), and $A=0$ iff~$\tau(A)=0$. As a result $A\cdot B = 0$ iff $A=0$ or $B=0$. Hence~$M$ satisfies $A\cdot B = 0$ iff $B\cdot A = 0$, making $M$ 
        a SEA.
\end{example}
This example is interesting because of how close its structure is to that of a SEA coming from a C$^*$-algebra while still being different in crucial ways. First of all, it is convex and  it contains no non-zero infinitesimal elements: if $n A\in M$ for all $n$, then $A=0$. As far as we are aware, this is the first example of a non-commutative convex SEA without infinitesimal elements that is not somehow related to examples from quantum theory.
It might hence serve as a counter-example to reasonable sounding hypotheses or characterisations of (convex) SEAs.

\begin{remark}
An effect algebra that does not contain infinitesimal elements is often called \emph{Archimedean}. 
For a convex effect algebra $E\cong [0,1]_V$ there is however also a stronger notion that is referred to as being Archimedean. 
Namely, that for any element of the corresponding vector space $a\in V$, if $na\leq 1$ for all $n\in \N$ then $a\leq 0$. 
We will refer to a convex effect algebra that has this property as \emph{strongly Archimedean}.
Having no non-zero infinitesimals is equivalent to the \emph{order unit semi-norm} $\norm{a}:=\inf\{\lambda\in \R~;~-\lambda 1 \leq a \leq \lambda 1\}$ being a norm, while being strongly Archimedean also requires the positive cone of $V$ to be closed in this norm.
\end{remark}

Coming back to Example~\ref{ex:assoc-SEA}, let $W$ denote the space spanned by $M$ in $R$. 
Then $W$ is a (non-Archimedean) ordered vector space where its order-unit semi-norm is in fact a norm. This norm satisfies $\norm{A} := \inf\{\lambda \in \R_{>0}~;~-\lambda \id \leq A\leq \lambda \id\} = \lvert \tau(A)\rvert$ and as a consequence we see that the sequential product is continuous in this norm.
This is significant because in Ref.~\cite{wetering2018sequential} it was shown that a convex finite-dimensional strongly Archimedean SEA where the sequential product is continuous in the norm must be order-isomorphic to the unit interval of a \emph{Euclidean Jordan algebra}, a type of structure closely related to C$^*$-algebras, and hence quantum theory.
It was already noted in~\cite{wetering2018sequential} that this result does not hold if the SEA contains infinitesimal elements, but it was left open whether the stronger notion of being Archimedean was necessary, or if having no non-zero infinitesimals was sufficient.
As Example~\ref{ex:assoc-SEA} is not order-isomorphic to the unit interval of a Euclidean Jordan algebra we indeed see that the mere absence of infinitesimal elements does not suffice for the result.

Example~\ref{ex:assoc-SEA} has some further noteworthy properties. Its product is non-commutative, but unlike the standard sequential product on for instance a C$^*$-algebra, is also associative. 
Finally, it is `close' to being directed complete, in the sense that any non-empty directed set $S$ has a minimal upper bound, but, unless $\bigvee \tau(S) = 1$, such upper bounds are not unique, and hence $S$ does not have a supremum.

\section{Boolean sequential effect algebras}\label{sec:boolean}
Naming sequential effect algebras where every element is idempotent `Boolean' of course suggests that such an effect algebra must actually be a Boolean algebra. This is indeed the case, and as far as we know has not been observed before. Let us therefore prove this before we continue on to our main results. Note that these results are a strict generalization of those in Ref.~\cite{tkadlec2008atomic}.

\begin{lemma}
\label{lem:amultbidempotent}
Let~$a$ and~$b$ be elements of a SEA such that~$a\mult b$ is idempotent.
Then~$a\mult b \leq b$.
Moreover,
if~$a\mult b^\perp$ is idempotent too,
then~$a\mult b = b \mult a$.
\end{lemma}
\begin{proof}
Since~$a\mult b\leq a$,
and~$a\mult b$ is an idempotent,
we have $a\mult b \commu a$
and $a\mult b = (a\mult b) \mult a$,
by 
point~\ref{prop:SEAbasicproperties-4} of
Proposition~\ref{prop:SEAbasicproperties}.
Then, using~S4, we see that
    $a\mult b = (a\mult b)\mult (a\mult b)
    = ((a\mult b)\mult a)\mult b
    = (a\mult b) \mult b$,
    and so~$a\mult b\leq b$,
    by point~\ref{prop:SEAbasicproperties-4} of
Proposition~\ref{prop:SEAbasicproperties}.

Now suppose that~$a\mult b^\perp$ is an idempotent too.
    Since then~$a\mult b^\perp \leq b^\perp$,
    we have $b\mult (a\mult b^\perp)=0$,
    and so~$b \mult (a\mult b)= b\mult a$.
On the other hand, $a\mult b = b\mult (a\mult b)$,
because~$a\mult b\leq b$,
so altogether we get $a\mult b = b\mult (a \mult b) = b\mult a$.
\end{proof}
\begin{corollary}
\label{cor:booleans-central}
A Boolean idempotent of a SEA is central.
\end{corollary}
\begin{proof}
Let $E$ be a SEA and let $a\in E$
be a Boolean idempotent~$a$.
Then $a$ commutes with every~$b\in E$
    by Lemma~\ref{lem:amultbidempotent},
    because~$a\mult b$ and~$a\mult b^\perp$
    are idempotents,
    on account of being
     below the Boolean idempotent~$a$.
\end{proof}
%

\begin{proposition}
    Let $p$ and~$q$ be idempotents of a SEA~$E$. Then $p\mult q$ is an idempotent if and only if $p$ and $q$ commute.
    Moreover,
    in that case $p\mult q$ is the infimum of~$p$ and~$q$ in~$E$.
\end{proposition}
\begin{proof}
    If $p$ and $q$ commute, then it is straightforward to show that $(p\mult q)^2 = p\mult q$.

Now for the other direction, suppose $p\mult q$ is idempotent.
Then~$p\mult q^\perp $ is idempotent too,
and so~$p$ and~$q$ commute, by Lemma~\ref{lem:amultbidempotent}.

Finally, we must show that~$p\mult q$ is the infimum of~$p$ and~$q$ in~$E$.
By Lemma~\ref{lem:amultbidempotent}
we know that~$p \mult q$ is a lower bound of~$p$ and~$q$.
To show that~$p\mult q$ is the greatest lower bound,
let~$r\leq p,q$ be given.
Since~$p$ and~$q$ are idempotents,
we have $r\mult p = r = r\mult q$
and~$r\commu p$,
    so $r\mult(p\mult q) = (r\mult p)\mult q
    = r\mult q = r$.
    It follows that $r\leq p\mult q$, and so we conclude that indeed 
    $p\mult q = p\wedge q$.
\end{proof}

Recall that a SEA is Boolean when every element is idempotent.

\begin{proposition}\label{prop:SEAsharpisBoolean}
    Let $E$ be a Boolean SEA.
    Then $E$ is a Boolean algebra.
    Furthermore, for all $a,b \in E$, $a\mult
    b = a\wedge b$. If $E$ is normal, then~$E$ is complete as a Boolean
    algebra.
\end{proposition}
\begin{proof}
    Let $a,b\in E$. As $a$, $b$ and $a\mult b$ are idempotent by
    assumption, by the previous proposition they commute and $a\mult
    b  = a\wedge b$. We conclude that $E$ is commutative, and hence
    it is a (directed-complete) Boolean effect monoid. Ref.~\cite[Proposition~47]{first} then shows that it is a (complete) Boolean algebra.
\end{proof}

Using our later results it will turn out that any SEA containing only a finite number of elements must be a Boolean algebra (see Corollary~\ref{cor:finite-Boolean}).

\section{Almost-convex sequential effect algebras}\label{sec:almostconvex}

In this section we study what we call \emph{almost-convex} SEAs (abbreviated to a-convex). We do this because they naturally arise in a structure theorem for normal SEAs that we prove at the end of this section.

The content of this section is rather technical, with items~\ref{def:floor}--\ref{prop:a-convex-from-phi} serving as preparation for the main results: Theorem~\ref{thm-a-convex-thm} that characterises convex normal SEAs among a-convex normal SEAs in several different ways, and Theorems~\ref{thm:a-convexthm} and~\ref{thm:SEAsplitupinconvexandsharp} that together show that any normal SEA splits up into a direct sum of a Boolean algebra and an a-convex effect algebra.

In the next section we will show that an a-convex normal SEA factors
into a convex normal SEA and a `purely a-convex' normal SEA 
(Proposition~\ref{prop:a-convexsplitinconvexandsharp}).
\begin{definition}\label{def:almostconvex}
An \Define{a-convex action} (\emph{almost-convex action}) on an effect algebra~$E$
  is a map $\cdot: [0,1]\times E\rightarrow E$,
    where $[0,1]$ is the standard real unit interval, satisfying
    the following axioms for all $a,b\in E$ and~$\lambda, \mu \in
    [0,1]$:
    \begin{itemize}
        \item $\lambda\cdot (\mu\cdot a) = (\lambda\mu)\cdot a$.
    \item If $\lambda+\mu \leq 1$, then $\lambda \cdot a \perp
    \mu \cdot a$ and $\lambda \cdot a \ovee \mu \cdot a =
    (\lambda+\mu)\cdot a$.
        \item $1\cdot a = a$.
    \end{itemize}
  An \Define{a-convex effect algebra}
is an effect algebra~$E$ endowed
with such an a-convex action.
  When we say that an effect algebra~$E$ is a-convex
we mean that there's at least one a-convex action on~$E$
  (but note that in general, there might be more then one a-convex action on a given effect algebra.)
\end{definition}
\begin{remark}
Recall from Definition~\ref{def:convex} that a convex action on an effect algebra
  is an a-convex action that in addition satisfies
  the requirement that
for all summable $a$ and $b$, and $\lambda\in [0,1]$
    \begin{equation*}
    \lambda\cdot (a\ovee b) \ = \ \lambda\cdot a \ovee \lambda\cdot b.
    \end{equation*}
\end{remark}

The definition of a-convexity is strictly weaker then that of convexity. Before we demonstrate this with an explicit example, we first recall the following method for constructing new (sequential) effect algebras 
by taking the `disjoint union' of effect algebras:

\begin{definition}
  \label{def:horizontal-sum}
    Let $I$ be some indexing set, and let~$E_\alpha$ be an effect algebra for each $\alpha\in I$. The \Define{horizontal
    sum} \cite{foulis1994effect} of the $E_\alpha$ is then defined as $\HS(E_\alpha, I) :=
    \left(\coprod_{\alpha\in I} E_\alpha\right)/_{\sim}$, the
    disjoint union modulo the identification of all the zeros and
    and all the ones. 
    Explicitly, denoting an element of $\coprod_{\alpha\in I} E_\alpha$ by $(a,\alpha)$ where $\alpha \in I$ and $a\in E_\alpha$ we set $(a,\alpha) \sim (b,\beta)$ iff $a=b=1$ or
    $a=b=0$ or $a=b$ and $\alpha=\beta$. 
    We define the complement as $(a,\alpha)^\perp = (a^\perp, \alpha)$. 
    The summability relation is defined as follows. For elements $(a,\alpha)$ and $(b,\beta)$ we set $(a,\alpha)\perp (b,\beta)$ iff $\alpha=\beta$ and $a\perp b$, or if $a=0$ or $b=0$. The sum $(a,\alpha)\ovee (b,\beta)$ for summable elements $(a,\alpha)$ and $(b,\beta)$ is defined by case distinction as follows. If $a=0$, then $(a,\alpha)\ovee (b,\beta) = (b,\beta)$. If $b=0$, then $(a,\alpha)\ovee (b,\beta) = (a,\alpha)$. If $a\neq 0$ and $b\neq 0$ then necessarily $\alpha=\beta$ for summable elements and we set $(a,\alpha)\ovee (b,\alpha) = (a\ovee b,\alpha)$.
    This complement, summability relation, and addition makes the horizontal sum an
    effect algebra.
\end{definition}
This definition might look a bit arbitrary, but note that the horizontal sum of two effect algebras is in fact the coproduct in the category of effect algebras with unital morphisms~\cite{jacobs2012coreflections}.
The horizontal sum of two sequential effect algebras does not have to again be a sequential effect algebra. Necessary and sufficient conditions were found in Ref.~\cite[Theorem 8.2]{gudder2002sequential} for a horizontal sum of sequential effect algebras to allow a sequential product. Let us give a simple example of such a SEA coming from a horizontal sum.
\begin{example}\label{ex:horizontal-interval}
    Let $H$ be the horizontal sum of the unit interval with
    itself. I.e.\ $H$ is the disjoint union of the unit interval
    with itself, that we will call the left~$[0,1]_L$ and the
    right~$[0,1]_R$ interval, where $0_L = 0_R$ and $1_L=1_R$ are
    identified.  This is an effect algebra where addition is only
    defined when elements are both from the left respectively the
    right interval, and the complement is $\lambda_L^\perp =
    (1-\lambda)_L$ (same for right). It is easy to see that this
    effect algebra is directed complete (as each of the unit intervals is). It is also a normal SEA
    with the product $\lambda_A\mult \mu_B = (\lambda\mu)_A$ where
    $A,B\in \{L,R\}$.  We can give $H$ two different a-convex
    structures, determined by either setting $\lambda\cdot 1 =
    \lambda_L$ or $\lambda\cdot 1 = \lambda_R$. For every other
    element there is a unique choice given by $\lambda\cdot \mu_A
    = (\lambda\mu)_A$ for $A\in \{L,R\}$.
    It is straightforward to check that either choice for $\lambda\cdot
    1$ gives an a-convex action on $H$. This however does not
    make $H$ a convex effect algebra, because (supposing without
    loss of generality that $\lambda\cdot 1 = \lambda_L$) we get
    $\lambda \cdot (\frac12_R \ovee \frac12_R) = \lambda\cdot 1 =
    \lambda_L$, while $\lambda \cdot \frac12_R \ovee \lambda \cdot
    \frac12_R = (\lambda \frac12)_R \ovee (\lambda \frac12)_R =
    \lambda_R$.
\end{example}

\noindent On a normal SEA~$E$ an a-convex action
    yields an additive map~$\varphi\colon [0,1] \to E$
            given by~$\varphi(\lambda) = \lambda\cdot 1$.
We will prove its converse, but that will require some preparation.

Recall that a normal SEA also has infima of decreasing sequences (Remark~\ref{remark:infima}).

\begin{definition}\label{def:floor}
    Let $E$ be a normal SEA. For $a\in E$ we define its \Define{floor} to be $\floor{a} := \bigwedge_n a^n$.
\end{definition}

\begin{lemma}
    Let $E$ be a normal SEA. For $a\in E$, we have $\floor{a} \leq a$. Furthermore, $\floor{a}$ is the largest idempotent below $a$.
\end{lemma}
\begin{proof}
  To start, 
  since~$a$ commutes with each~$a^n$,
  we have~$a \commu  \floor{a} = \bigwedge_n a^n$,
  and $\floor{a}\mult a = a\mult \floor{a} = \bigwedge_n a\mult a^n 
  = \bigwedge_n a^{n+1}
    = \bigwedge_n a^n = \floor{a}$.
Then~$\floor{a}\mult a^n = \floor{a}$ for all $n$,
  and so $\floor{a}^2 = \bigwedge_n \floor{a} \mult a^n = \floor{a}$.
Hence~$\floor{a}$ is an idempotent.

Let~$p$ be an idempotent below~$a$;  we must show that~$p\leq \floor{a}$.
Since~$p\leq a$, we have $p\mult a = a\mult p=  p$.
Note that~$p\mult a^2 = p\mult(a\mult a)=(p\mult a)\mult a=p\mult a=p$,
using here that~$a\commu p$. By a similar reasoning
we get $p\mult a^n =p$ for all~$n$.
  Thus~$p\mult \floor{a} = \bigwedge_n p\mult a^n 
	= p$, and so $p\leq \floor{a}$.
\end{proof}

\begin{lemma}\label{lem:uniquedivision}
Given an element~$a$ of a normal SEA~$E$ with~$\floor{a}=0$
and~$n\in \N_{>0}$,
there is a unique~$a'\in E$ with~$a=na'$.
  Moreover, $a'\in \{a\}''$.
\end{lemma}
\begin{proof}
Concerning uniqueness, suppose for now that there is an~$a'\in \{a\}''$
with~$na'=a$, and let~$b\in E$ with~$a=nb$ be given.
Since~$b\commu nb=a$, we have~$b\in \{a\}'$,
  and thus~$b\commu a'$, using here that~$a'\in \{a\}''$.
It follows that~$b$ and~$a'$ are both part
  of the commutative normal SEA  $\{b,a'\}''$.
Now,
the representation theory for directed-complete effect monoids
(cf.~Section~\ref{sec:representation-theorem})
gives us that
  $nb=na'$ implies~$b=a'$ (indeed, this implication
  holds for commutative $C^*$-algebras,
  and trivially for Boolean algebras). Hence,~$a'$ is unique.

For the existence of~$a'\in\{a\}''$ 
we can assume without loss of generality that~$E=\{a\}''$,
and, moreover,
that~$E = B\oplus C$
for some complete Boolean algebra~$B$,
and some convex normal effect monoid~$C$,
again by the representation theory for directed-complete effect monoids.
  Then~$a=(b,c)$ for some~$b\in B$ and~$c\in C$.
  We claim that~$b=0$.
  For this, note that~$b$ must be an idempotent since~$B$ is Boolean, 
  and so~$\floor{b}=b$.
Since $0=\floor{a}=\floor{(b,c)}=(\floor{b},\floor{c})=(b,\floor{c})$,
we have~$b=0$.
Finally, define~$a':=(0,\frac{1}{n}\cdot c)$,
  using here that~$C$ is convex, and observe that~$na'=a$.
\end{proof}

Note that the condition~$\floor{a} = 0$
    is necessary in the previous Lemma. Indeed, in
    Example~\ref{ex:horizontal-interval} taking $a=1$, we get $2(\frac12_L) = a
    = 2(\frac12_R)$ while $\frac12_L\neq \frac12_R$.
\begin{proposition}
\label{prop:add-into-nsea}
Let~$\varphi\colon [0,1]\to E$
be an additive map into a normal SEA~$E$.
\begin{enumerate}
\item
\label{prop:add-into-nsea-1}
$\varphi$ is normal:
we have $\varphi(\bigvee D)=\bigvee_{\lambda\in D} \varphi(\lambda)$
for every directed subset~$D$ of~$[0,1]$.
\item
\label{prop:add-into-nsea-2}
$\varphi(\lambda)\commu \varphi(\mu)$
        for all~$\lambda,\mu\in [0,1]$.
\item
\label{prop:add-into-nsea-5}
    $\floor{\varphi(\lambda)} = 0 $ for any~$\lambda \in [0,1)$.
\item
\label{prop:add-into-nsea-4}
If~$\varphi(\lambda) = \psi(\lambda)$
    for some~$\lambda \in (0,1)$
        and additive~$\psi\colon [0,1] \to E$,
        then~$\varphi = \psi$.
\end{enumerate}
\end{proposition}
\begin{proof}
\ref{prop:add-into-nsea-1}.
Since
$\bigvee_{\lambda\in D} \varphi(\lambda)
\leq 
\varphi(\bigvee D)$
the difference
$\varphi(\bigvee D)
 \,\ominus\, \bigvee_{\lambda\in D} \varphi(\lambda)$ 
exists;
we must show that it is zero.
Let a natural number~$n>0$ be given,
    and pick~$\mu\in D$ with $\bigvee D- \mu \leq \frac{1}{n}$. 
  Since
\begin{equation*}
    \textstyle
\varphi(\bigvee D)
 \,\ominus\, \bigvee_{\lambda\in D} \varphi(\lambda)
\ \leq\ 
\varphi(\bigvee D)
    \,\ominus\, \varphi(\mu)
\ =\ 
\varphi(\bigvee D-\mu)
    \ \leq \ \varphi(\frac{1}{n}),
\end{equation*}
using that~$\varphi(\mu)\leq \bigvee_{\lambda\in D} \varphi(\lambda)$
in the first inequality,
we see then that
$\varphi(\bigvee D)
 \,\ominus\, \bigvee_{\lambda\in D} \varphi(\lambda)$
    has an $n$-fold sum for all $n$,
    and hence must therefore be zero by Lemma~\ref{lem:archemedeanomegadirectedcomplete}.

\ref{prop:add-into-nsea-2}.
Let~$n,m>0$ be natural numbers.
    Since~$\varphi(\frac{1}{nm})$
    commutes with itself,
    it commutes with $n\varphi(\frac{1}{nm}) = \varphi(\frac{1}{m})$
    by~S5. But then $m\varphi(\frac{1}{nm}) = \varphi(\frac{1}{n})$
    commutes with~$\varphi(\frac{1}{m})$ too,
    again by~S5.
Going on like this
we see that~$\varphi(\frac{k}{n})\commu \varphi(\frac{\ell}{m})$ 
for all natural numbers $k\leq n$ and~$\ell \leq m$.
Whence~$\varphi(q)\commu \varphi(r)$ for all \emph{rational}
$q,r\in [0,1]$.
Now let~$x,y\in [0,1]$ be arbitrary,
    and pick directed sets~$C,D\subseteq[0,1]$  of rational numbers with $\bigvee C=x$
and $\bigvee D = y$.
Then~$\varphi(c)\commu \varphi(d)$ for all~$c\in C$ and~$d\in D$,
and so~$\varphi(x)=\bigvee_{c\in C} \varphi(c)\commu \varphi(d)$
for all~$d\in D$, by~S6, and the fact that~$\varphi$ is normal.
But then~$\varphi(x)\commu \bigvee_{d\in D}\varphi(d)=\varphi(y)$ too.


\ref{prop:add-into-nsea-5}.
Pick a natural number~$n >0$ with~$\lambda \leq 1 - \frac{1}{n}$.
Then~$\floor{\varphi(\lambda)} \leq \floor{\varphi(1 - \frac{1}{n})}$
    and so it suffices to show that~$\floor{\varphi(1-\frac{1}{n})}=0$.
To this end let~$p$ be an idempotent with~$p \leq \varphi(1 - \frac{1}{n})$.
We have to show that~$p=0$.
Since $p\leq \varphi(1-\frac{1}{n})=\varphi(1)\ominus \varphi(\frac{1}{n})
\leq 1\ominus \varphi(\frac{1}{n})$,
and so~$\varphi(\frac{1}{n})\leq p^\perp$,
we have~$\varphi(1) = n\varphi(\frac{1}{n})\leq p^\perp$
by Lemma~\ref{lem:summableunderidempotent},
and whence~$p\leq \varphi(1)^\perp$.
On the other hand $p\leq \varphi(1-\frac{1}{n})\leq \varphi(1)$,
so $p$, being below both $\varphi(1)$ and $\varphi(1)^\perp$,
is summable with itself.  
Since~$p$ is an idempotent,
we get $p\ovee p\leq p$ by Lemma~\ref{lem:summableunderidempotent},
and so~$p=0$, as desired.

\ref{prop:add-into-nsea-4}.
Since~$\varphi(\lambda)=\psi(\lambda)$,
we have $\varphi(\frac{\lambda}{n})=\psi(\frac{\lambda}{n})$
for all natural numbers~$n>0$,
by Lemma~\ref{lem:uniquedivision},
using here that~$\floor{\varphi(\lambda)} = \floor{\psi(\lambda)}=0$
by point~\ref{prop:add-into-nsea-5}.
Since numbers of the form~$\frac{m\lambda}{n}$
lie dense in~$[0,1]$,
point~\ref{prop:add-into-nsea-1}
entails that~$\varphi=\psi$.
\end{proof}

\begin{definition}
    Let $E$ be an effect algebra. We call an element $h\in E$ a \Define{half} when $h\ovee h = 1$.
\end{definition}
\begin{proposition}\label{prophalvecentral}
Let~$E$ be a normal SEA.
A half is central iff it is unique.
\end{proposition}
\begin{proof}
Given elements~$h$ and~$g$  of~$E$
with $h\ovee h=1=g\ovee g$,
and~$h\mult g = g\mult h$,
we have
\[h=h\mult 1
  = h\mult (g\ovee g)
  = h\mult g\ovee h\mult g
  = g\mult h\ovee g\mult h
  = g\mult(h\ovee h) = g,\]
so commuting halves are equal.
In particular,
a half is unique when it is central.

For the converse,
suppose that~$h$ is the only half in~$E$,
and let~$a$ be an element of~$E$.
We must show that~$a$ and~$h$ commute.
To this end,
note that~$a \mult h \ovee a^\perp \mult h$
is a half,
and so~$h=a\mult h\ovee a^\perp\mult h$,
 by uniqueness of~$h$.
Since~$a\mult h$ commutes with $a\mult h \ovee a\mult h =  a$,
and, similarly, $a^\perp \mult h$ commutes with~$a^\perp$,
and thus with~$a$ too, 
  we see that~$a$ commutes with~$a\mult h\ovee a^\perp \mult h = h$.
Whence~$h$ is central.
\end{proof}

\begin{proposition}
\label{prop:a-convex-from-phi}
Let~$E$ be a normal SEA.
Any unital, additive map $\varphi\colon [0,1]\to E$
gives an a-convex action~$\cdot_\varphi$
on~$E$ via $\lambda\cdot_\varphi a = a\mult \varphi(\lambda)$.
Moreover,
for an a-convex action~$\cdot$, 
an a-convex action~$\cdot$
is of this form iff
  $\lambda \cdot a = a\mult (\lambda\cdot 1)$
  for all~$a\in E$ and~$\lambda\in [0,1]$.
\end{proposition}
\begin{proof}
\emph{($\cdot_\varphi$ is an a-convex action)}
    Clearly~$1 \cdot_\varphi a = a$
        and~$(\lambda + \mu)\cdot_\varphi a 
        = \lambda \cdot_\varphi a \ovee \mu \cdot_\varphi a$
            for~$\lambda + \mu \leq 1$ and~$a\in E$.
The only difficulty here is in establishing the last remaining
  condition, that~$\mu \cdot_\varphi (\lambda \cdot_\varphi a) =
    (\mu\lambda)\cdot_\varphi a$
    given~$\mu,\lambda\in[0,1]$ and~$a\in E$.
  Since for fixed~$\lambda$ and~$a$ 
  both~$\mu \mapsto \mu \cdot_\varphi (\lambda \cdot_\varphi a)$
and~$\mu \mapsto (\mu \lambda) \cdot_\varphi a$
give additive maps~$[0,1] \to E$,
it suffices 
by point~\ref{prop:add-into-nsea-4} of Proposition~\ref{prop:add-into-nsea}
to show
  that~$\frac12 \cdot_\varphi (\lambda \cdot_\varphi a) 
  = (\frac12 \lambda) \cdot_\varphi a$.
For~$\lambda=1$, this is obvious enough,
and for~$\lambda<1$
this follows immediately from the fact that~$\lambda\cdot_\varphi a$
has a unique half
by Lemma~\ref{lem:uniquedivision},
because $\floor{\lambda \cdot_\varphi a} =0$
by point~\ref{prop:add-into-nsea-5} of Proposition~\ref{prop:add-into-nsea}.
Note that we could not 
  prove~$\frac12 \cdot_\varphi (\lambda \cdot_\varphi a) 
  = (\frac12 \lambda) \cdot_\varphi a$
  here
  by simply
  applying point~\ref{prop:add-into-nsea-4} of
  Proposition~\ref{prop:add-into-nsea} again
  as it's a priori not clear that 
  $\lambda \mapsto \frac12\cdot_\varphi(\lambda\cdot_\varphi a)$ is additive.

\emph{(The condition $\lambda \cdot a = a \mult (\lambda \cdot 1)$)}
  Note that~$a\mult (\lambda\cdot_\varphi 1)
  = a\mult (1\mult \varphi(\lambda)) = a\mult \varphi(\lambda)=
  \lambda\cdot_\varphi a$ for any unital, additive~$\varphi$.
  Conversely, for an a-convex action $\cdot$ 
with $\lambda\cdot a = a\mult(\lambda \cdot 1)$
  for all~$\lambda\in[0,1]$ and~$a\in E$,
  we immediately get~$\cdot = \cdot_\varphi$ where
  $\varphi(\lambda)=\lambda\cdot 1$.
\end{proof}
One might think that the process of constructing an a-convex action from a unital additive map, and retrieving a unital additive map from an a-convex action are each others inverse, i.e.~that we always have the condition $\lambda\cdot a = a\mult(\lambda\cdot 1)$ stated in the proposition above. Example~\ref{ex-aconvex-not-determined-by-scalars} below
shows that this is not the case.
For convex actions, however,
the equation $\lambda\cdot a=a\mult(\lambda\cdot 1)$ 
does hold
(because there is only one a-convex action
in the presence of a convex action, see Theorem~\ref{thm-a-convex-thm}).
\begin{example}
  \label{ex-aconvex-not-determined-by-scalars}
    Let $H=HS([0,1],[0,1])$ be the horizontal sum of $[0,1]$ with itself from Example~\ref{ex:horizontal-interval}, equipped with the a-convex action determined by $\lambda\cdot 1 = \lambda_L$.
    Now let $E=HS(H\oplus H, [0,1])$, the horizontal sum of $H\oplus H$ and $[0,1]$. We define a sequential product by case distinction: Let $(a,b),(a',b')\in H\oplus H$ and take $\mu,\mu' \in[0,1]$. We set $\mu\mult \mu' = \mu\mu'$ and $\mu\mult (a,b) = \frac12 (\mu a + \mu b)$, where we interpret $a$ and $b$ as elements of $[0,1]$. We set $(a,b)\mult (a',b') = (a\mult a', b\mult b')$ and $(a,b)\mult \mu = (a\mult \mu_R,b\mult \mu_R)$ where we interpret $\mu$ as $\mu_R$ in $H$. So in particular $(1,0)\mult \mu = (\mu_R,0)$.

    Now we equip $E$ with an a-convex action. For $(a,b)\in H\oplus H$ with $(a,b)\neq (1,1)$ we set $\lambda\cdot (a,b) = (\lambda\cdot a,\lambda\cdot b)$, and for $\mu\in[0,1]$ with $\mu\neq 1$ we set $\lambda\cdot \mu = \lambda\mu$. Finally, for $1\in E$ we set $\lambda\cdot 1 = \lambda_R$, i.e.~an element of $[0,1]$. We now see that $\lambda\cdot (1,0) = (\lambda\cdot 1,0) = (\lambda_L,0)$, while $(1,0)\mult (\lambda\cdot 1) = (1,0)\mult \lambda = (1\mult \lambda_R,0) = (\lambda_R,0)$.
\end{example}

\begin{theorem}\label{thm-a-convex-thm}
For a normal SEA~$E$ the following are equivalent.
\begin{enumerate}
\item
\label{thm-convex-convex}
There is a convex action on~$E$.
\item
\label{thm-convex-unique-a-convex}
There is precisely one a-convex action on~$E$.
\item
\label{thm-convex-central-half}
There is a central~$h\in E$ with $h\ovee h=1$.
\item
\label{thm-convex-unique-half}
There is precisely one $h\in E$ with $h\ovee h = 1$.
\item
\label{thm-convex-unique-halves}
For each~$a\in E$ there is precisely one~$b\in E$
with $b\ovee b = a$.
\item
\label{thm-convex-restriction-to-center}
There is an a-convex action on~$E$,
and any such a-convex action~$\cdot$ 
can be restricted to~$Z(E)$
in the sense that  $\lambda\cdot a$ is central
  for all~$\lambda\in[0,1]$ and~$a\in Z(E)$.
\item
\label{thm-convex-center-convex}
$Z(E)$ is convex.
\end{enumerate}
Moreover, in that case
  $\lambda\cdot a = a\mult(\lambda\cdot 1)$
for the unique a-convex action~$\cdot$ on~$E$,
and all~$a\in E$ and~$\lambda\in [0,1]$.
\end{theorem}
\begin{proof}
We begin by proving that
  \ref{thm-convex-convex}--\ref{thm-convex-unique-halves}
  are equivalent by producing the loop of implications
\ref{thm-convex-convex}$\Rightarrow$%
\ref{thm-convex-unique-half}$\Rightarrow$%
\ref{thm-convex-central-half}$\Rightarrow$%
\ref{thm-convex-unique-halves}$\Rightarrow$%
\ref{thm-convex-unique-a-convex}$\Rightarrow$%
\ref{thm-convex-convex}.
Along the way, we establish 
\ref{thm-convex-unique-a-convex}$\Rightarrow$\ref{thm-convex-unique-half}
as it is needed to prove
  \ref{thm-convex-unique-a-convex}$\Rightarrow$\ref{thm-convex-convex}.

  \emph{(\ref{thm-convex-convex}$\Rightarrow$\ref{thm-convex-unique-half})}
Let~$\cdot$ be a convex action on~$E$.
Then~$1$ has a unique half:
 if~$h \ovee h = 1$ for some~$h\in E$,
    then~$h = (\frac{1}{2}\ovee\frac{1}{2})\cdot h
    =\frac{1}{2} \cdot h \ovee \frac{1}{2}\cdot h
        = \frac{1}{2}\cdot (h \ovee h) = \frac{1}{2} \cdot 1$.

\emph{(\ref{thm-convex-unique-half}$\Rightarrow$\ref{thm-convex-central-half})}
Follows immediately from Proposition~\ref{prophalvecentral}.

\emph{(\ref{thm-convex-central-half}$\Rightarrow%
$\ref{thm-convex-unique-halves})}
Let~$a\in E$ be given; we must show that~$a$ has a unique half.
Let~$h$ be a central element such that~$h\ovee h=1$.
Since~$b= b\mult (h\ovee h)
  = (b\mult h) \ovee (b\mult h)
  = (h\mult b) \ovee (h\mult b)
  = h \mult (b\ovee b) = h\mult a$
  for any~$b\in E$ with~$a=b\ovee b$,
  we see that~$a$ has a unique half.

\emph{(\ref{thm-convex-unique-halves}$\Rightarrow%
$\ref{thm-convex-unique-a-convex})}
Assume each~$a\in E$ has a unique half.
For uniqueness,
let~$\cdot_1$ and~$\cdot_2$ be a-convex actions on~$E$.
To prove that~$\cdot_1=\cdot_2$,
we must show that~$(\,\cdot\,)\cdot_1 a
= (\,\cdot\,)\cdot_2 a$ for given~$a\in E$,
and for this it suffices to
show that~$\frac{1}{2} \cdot_1 a = \frac{1}{2}\cdot_2 a$,
by Proposition~\ref{prop:add-into-nsea}.\ref{prop:add-into-nsea-4}.
But since both $\frac{1}{2}\cdot_1 a$ and~$\frac{1}{2}\cdot_2 a$
are halves of~$a$, this follows by assumption.

It remains to be shown that there is at least one a-convex action on~$E$.
To this end, let~$h$ be a half of~$1$, and note that~$\{h\}''$
being a directed-complete effect monoid with a half
is isomorphic to $[0,1]_{C(X)}$ for some extremally-disconnected
compact Hausdorff space~$X$, via some isomorphism
  $\Phi\colon [0,1]_{C(X)}\to \{h\}''$.
The assignment $\lambda\mapsto \Phi(\lambda\mathbbm{1})$,
  where~$\mathbbm{1}$ is the function on~$X$ that is constant one,
  gives a unital, additive map~$\varphi\colon [0,1]\to E$,
  and so~$E$ has a a-convex action given by~$\lambda\cdot a
  = a\mult \varphi(\lambda)$ by
  Proposition~\ref{prop:a-convex-from-phi}.

\emph{(\ref{thm-convex-unique-a-convex}$\Rightarrow%
$\ref{thm-convex-unique-half})}
Suppose that~$E$ has a unique a-convex action~$\cdot$.
Then clearly~$1$ has a half given by~$\frac{1}{2}\cdot 1$.
Concerning uniqueness, let~$h\in E$ with~$h\ovee h=1$ be given.
Considering the directed-complete effect monoid~$\{h\}''$
  we can find a unital, additive map~$\varphi\colon [0,1]\to E$
  with~$\varphi(\frac{1}{2}) = h$,
  which yields an a-convex action~$\cdot_\varphi$ on~$E$
  given by $\lambda\cdot_\varphi a= a\mult\varphi(\lambda)$ by Proposition~\ref{prop:a-convex-from-phi}.
Since there is only one a-convex action on~$E$,
we get~$\cdot=\cdot_\varphi$, and 
thus $\frac{1}{2}\cdot 1 = \frac{1}{2}\cdot_\varphi 1 =\varphi(\frac12)= h$.

\emph{(\ref{thm-convex-unique-a-convex}$\Rightarrow%
$\ref{thm-convex-convex})}
Let~$\cdot$ be the unique a-convex action on~$E$.
Given summable $a,b\in E$ we
must show that
$\lambda\cdot a\ovee \lambda\cdot b = \lambda\cdot(a\ovee b)$
for all~$\lambda\in [0,1]$.
Note that by Proposition~\ref{prop:add-into-nsea}.\ref{prop:add-into-nsea-4}
it suffices to show that
$\frac12 \cdot a \ovee \frac12 \cdot b = \frac12\cdot (a\ovee b)$.
Since both sides of this equation
are clearly halves of~$a\ovee b$,
we are done if halves are unique---which indeed they are,
since we have already established that
\ref{thm-convex-unique-a-convex}$\Rightarrow$%
\ref{thm-convex-unique-half}$\Rightarrow$%
\ref{thm-convex-central-half}$\Rightarrow$%
\ref{thm-convex-unique-halves}.

\emph{Whence~
\ref{thm-convex-convex},
\ref{thm-convex-unique-a-convex},
\ref{thm-convex-central-half},
\ref{thm-convex-unique-half}, and 
\ref{thm-convex-unique-halves}
are equivalent.}
\noindent
We continue by showing that
\ref{thm-convex-restriction-to-center}
is equivalent to
\ref{thm-convex-convex}--\ref{thm-convex-unique-halves}.

\emph{(\ref{thm-convex-restriction-to-center}%
$\Rightarrow$%
\ref{thm-convex-central-half})}
Let~$\cdot$ be an a-convex action on~$E$.
Then, clearly, $h:=\frac{1}{2}\cdot 1$ 
is a central element
that obeys~$h\ovee h=1$.

\emph{(\ref{thm-convex-convex}--\ref{thm-convex-unique-halves}%
$\Rightarrow$\ref{thm-convex-restriction-to-center})}
Let~$\cdot$ be the unique a-convex action on~$E$ 
from~\ref{thm-convex-unique-a-convex}.
We must show that~$\cdot$ can be restricted to~$Z(E)$.
Since $\lambda\cdot'a=a \mult(\lambda \cdot a)$ defines
an a-convex action~$\cdot'$ 
on~$E$, and~$\cdot$ is unique, we have $\cdot=\cdot'$,
and so~$\lambda\cdot a = a\mult(\lambda\cdot 1)$
for all~$a\in E$ and~$\lambda\in [0,1]$.
Thus to prove that~$\cdot$ restricts to~$Z(E)$,
that is, $\lambda\cdot a= a\mult (\lambda \cdot 1 )\in Z(E)$
for all~$a\in Z(E)$ and~$\lambda\in [0,1]$,
it suffices to show that~$\lambda\cdot 1$ is central
for all~$\lambda\in [0,1]$.
By~\ref{thm-convex-central-half} and~\ref{thm-convex-unique-half}
we already know that~$\frac{1}{2}\cdot 1$ is central.
Note that~$\{\frac{1}{2}\cdot 1\}'=E$
 (because~$\frac12\cdot 1$ is central), and
so~$\{\frac{1}{2}\cdot 1\}''=Z(E)$.
It thus suffices to show that~$\lambda\cdot 1 \in \{\frac12\cdot 1\}''$.
From before we know that~$\{\frac12\cdot 1\}''$
is a directed-complete effect monoid and hence has its own set
of scalars that can be parametrised by a unital, additive
map~$\varphi\colon [0,1]\to E$ with~$\varphi(\frac12)=\frac12\cdot 1$.
Since~$\lambda\mapsto \lambda\cdot 1$ is an additive map too
that coincides with~$\varphi$ on~$\frac12$,
we get~$\varphi=(\,\cdot\,)\cdot 1$
by Proposition~\ref{prop:add-into-nsea}.\ref{prop:add-into-nsea-4},
and so~$\lambda\cdot 1\in Z(E)$ for all~$\lambda\in [0,1]$.

Finally,
\emph{(\ref{thm-convex-restriction-to-center}\,\emph{\&}\ref{thm-convex-convex}%
$\Rightarrow$\ref{thm-convex-center-convex})} 
and 
\emph{(\ref{thm-convex-center-convex}%
$\Rightarrow$\ref{thm-convex-convex})} are obvious,
so
\ref{thm-convex-convex}--\ref{thm-convex-center-convex}
are equivalent.
\end{proof}



\begin{theorem}\label{thm:a-convexthm}
Let~$E$ be a normal SEA.
Let $S\subseteq E$ be the set of idempotents $p$ for which~$p \mult E$ 
  can be equipped with an a-convex action.
  Then $S$ has a maximal element $p_0$
  and this maximal element is central
        and~$p_0^\perp \mult E$ is Boolean.
\end{theorem}
\begin{proof}
We begin by showing that~$E$ has a maximal self-summable element
(cf.~\cite[Lemma~56]{first}), using Zorn's lemma.
For this we must show 
that every chain~$D$ in~$A :=\{ a \in E;\  a \perp a \}$
has an upper bound.
If~$D$ is empty, then~$0 \in A$ is clearly an upper bound, so
we may assume~$D$ is not empty.
Being a non-empty chain, $D$ is directed,
and we may define~$u := \bigvee D$.
We will show that~$u$ is self-summable, 
making it the upper bound of~$D$ we were looking for.

To see that~$u$ is self-summable,
first note that any $d,d'\in D$ are summable,
that is, $d'\leq d^\perp$.
Indeed, either  $d'\leq d$ and thus~$d'\leq d\leq d^\perp$ --- using
here that~$d$ is self-summable --- or 
$d\leq d'$ and thus~$d'\leq (d')^\perp \leq d^\perp$.
Thus as~$d'\leq d^\perp$ for all~$d,d'\in D$,
  we get $u=\bigvee D \leq \bigwedge_{d\in D} d^\perp
  = (\bigvee D)^\perp= u^\perp$,
and so~$u$ is self-summable.

Write~$a$ for the maximal self-summable element,
and define~$p :=  a \ovee a$.
To see that~$p$ is an idempotent,
it suffices to show that~$p^\perp \mult p=0$.
To this end we will prove that $(p^\perp)^2\mult p = 0$,
because then $(p^\perp\mult p)^2 = ((p^\perp)^2\mult p)\mult p = 0$,
which is enough, by Lemma~\ref{lem:nonilpotents}.
To prove $(p^\perp)^2\mult p =0$ we show that~$(p^\perp)^2\mult p \ovee a$ is self-summable,
since then~$(p^\perp)^2\mult p \ovee a \leq a$
(by maximality of~$a$,)
and hence~$(p^\perp)^2 \mult p = 0$.
Now, as $p^\perp\mult p$ is self-summable we have~$2(p^\perp \mult p)\leq 1$,
  and so by multiplying on the left with $p^\perp$ we get $2((p^\perp)^2\mult p) \leq p^\perp$
  and thus $2((p^\perp)^2\mult p)$ is summable with~$p := 2a$,
  which implies~$(p^\perp)^2\mult p \ovee a$ is self-summable.
  Hence~$p$ is an idempotent.

Next we will show that~$p^\perp$ is Boolean,
        i.e.~that every~$s\leq p^\perp$ is idempotent.
Given such~$s\leq p^\perp$, note
  that~$s \mult s^\perp$ is summable with itself (cf.~Lemma~\ref{lem:selfsummable}),
        and so~$s \mult s^\perp \ovee s \mult s^\perp \leq p^\perp
            = (a \ovee a)^\perp$
            by Lemma~\ref{lem:summableunderidempotent}.
    Thus~$a \ovee a \perp s \mult s^\perp \ovee s \mult s^\perp$
    and so~$a \ovee s\mult s^\perp$ is summable with itself.
    By maximality of~$a$ we must have~$s \mult s^\perp = 0$,
    which shows that~$s$ is indeed an idempotent.
    Thus~$p^\perp$ is Boolean
        and by Corollary~\ref{cor:booleans-central}
        both~$p$ and~$p^\perp$ are central.

To find an a-convex action on  $p\mult E$, we consider~$\{a\}''\subseteq p\mult E$,
which is clearly a convex directed-complete effect monoid,
whose scalars are parametrised by a
  unital additive map~$[0,1] \to \{a\}'' \subseteq p \mult E$,
which yields an a-convex action on~$p \mult E$ 
by Proposition~\ref{prop:a-convex-from-phi}.

Finally, to prove that~$p$ is indeed maximal among idempotents with this property,
let~$q \in E$ be an idempotent for which~$q \mult E$ carries an a-convex action.
As $p^\perp$ is central we have~$p^\perp \mult q = q\mult p^\perp \leq q$, 
and so we may interpret~$\frac{1}{2}\cdot (p^\perp \mult q)$
using the a-convex action on~$q\mult E$.
Since $p^\perp \mult q \leq p^\perp$, 
we also have $\frac{1}{2} (p^\perp \mult q) \leq p^\perp$,
and thus~$\frac12 (p^\perp\mult q)$, 
being below the Boolean idempotent~$p^\perp$,
must be an idempotent. 
Since the only self-summable idempotent is zero,
$\frac{1}{2} (p^\perp \mult q)=0$.
    Thus~$p^\perp \mult q = 0$,
    and~$q \leq p$ as desired.
\end{proof}

\begin{theorem}\label{thm:SEAsplitupinconvexandsharp}
    Let $E$ be a normal SEA. Then $E\cong E_1\oplus E_2$, where $E_1$ is a-convex and $E_2$ is a complete Boolean algebra.
\end{theorem}
\begin{proof}
    By Theorem~\ref{thm:a-convexthm} there is a central idempotent~$p$
    such that~$p\mult
    E$ is a-convex, and $p^\perp\mult E$ is Boolean.
    By Proposition~\ref{prop:SEAsharpisBoolean}, $p^\perp\mult E$
    is then a complete Boolean algebra. Since $p$ is central
        we have by Proposition~\ref{prop:central-splits} that $E\cong p\mult E\oplus
    p^\perp \mult E$.
\end{proof}

\begin{corollary}\label{cor:finite-Boolean}
    Let $E$ be a SEA with a finite number of elements. Then $E$ is a Boolean algebra.
\end{corollary}
\begin{proof}
Any directed subset $S\subseteq E$ is also finite, and hence
    contains its supremum. Thus $E$ is normal and
    hence splits up as $E\cong E_1\oplus E_2$ where $E_1$ is a-convex
    and $E_2$ is a Boolean algebra. If $E_1 \neq \{0\}$ then it
    would necessarily contains a continuum of elements contradicting
    the finiteness of~$E$. Thus~$E \cong E_2$ is a Boolean algebra.
\end{proof}

\section{Pure a-convexity}\label{sec:purea-convex}
In this section we will show that
  any a-convex normal SEA
  factors as a convex normal SEA and a \emph{purely a-convex} normal SEA,
  which combined with Theorem~\ref{thm:SEAsplitupinconvexandsharp} gives our main characterisation theorem for normal SEAs.
	Intuitively, a purely a-convex SEA is an a-convex SEA that does not contain any convex `parts'. The factorisation result of this section, Proposition~\ref{prop:a-convexsplitinconvexandsharp}, makes this more precise, stating that an a-convex SEA can be split up into its convex part and its purely a-convex part.
\begin{definition}
We say that a a-convex SEA~$E$
is \Define{purely a-convex} when $Z(E)$ is Boolean,
and
an \Define{a-convex factor} if $Z(E) = \{0,1\}$.
\end{definition}
Recall that a normal SEA~$E$ is convex iff $Z(E)$ is convex
(see Theorem~\ref{thm-a-convex-thm}).
\begin{proposition}\label{prop:a-convexsplitinconvexandsharp}
In a normal a-convex SEA~$E$
there is a central idempotent~$p$ such that
$p\mult E$ is convex
and $p^\perp \mult E$ is purely a-convex.
\end{proposition}
\begin{proof}
Obviously $Z(E)$ is a commutative normal SEA, and thus also a directed-complete effect monoid.
There is then an idempotent $p\in Z(E)$ such that $p\mult Z(E)$ is convex, 
while $p^\perp\mult Z(E)$ is a Boolean algebra. 
Since $Z(p^\perp \mult E) = p^\perp \mult Z(E)$ is a Boolean algebra,
  we immediately see that~$p^\perp \mult E$ is purely a-convex.
Further, since~$Z(p\mult E)=p\mult Z(E)$ is convex,
  $p$ has a half in~$Z(p\mult E)$,
 which is a central half in~$p\mult E$,
  and so~$p\mult E$ is convex
  by Theorem~\ref{thm-a-convex-thm}.
\end{proof}

\begin{corollary}
\label{cor-aconvex}
For a normal a-convex SEA~$E$ the following are equivalent.
\begin{enumerate}
\item
  \label{cor-aconvex-boolean}
    $E$ is purely a-convex,
    that is,
  $Z(E)$ is Boolean.
\item
  \label{cor-aconvex-corner}
The only central idempotent~$p$ in~$E$
for which~$p\mult E$ is convex is~$p=0$.
\end{enumerate}
\end{corollary}
\begin{proof}
\emph{(\ref{cor-aconvex-boolean}$\Rightarrow$%
\ref{cor-aconvex-corner})}
Let~$p$ be a central idempotent for which~$p\mult E$ is convex.
Then the half~$\frac{1}{2}\cdot p$ of~$p$
in~$p\mult E$
  being an element of~$Z(p\mult E)\subseteq Z(E)$
  is idempotent (since~$Z(E)$ is Boolean),
  and so~$p=0$.

\emph{(\ref{cor-aconvex-corner}$\Rightarrow$%
\ref{cor-aconvex-boolean})}
Let~$p$ be the central idempotent from 
Proposition~\ref{prop:a-convexsplitinconvexandsharp}
with~$p\mult E$ convex and~$p^\perp\mult E$ Boolean.
Then~$p=0$ by assumption,
and so~$E=p^\perp\mult E$ is purely a-convex.
\end{proof}

Combining previous results, we get the following representation theorem.
\begin{theorem}\label{thm:maintheorem}
Let~$E$ be a normal SEA.
There is a complete Boolean algebra~$B$,
    convex normal SEA~$E_c$
    and purely a-convex SEA~$E_{ac}$
    with~$E \cong B \oplus E_c \oplus E_{ac}$.
\end{theorem}
Boolean algebras are obviously very well-studied and understood,
and convex normal SEAs also seem to be quite well-behaved (they are for instance order-isomorphic to the unit interval
    of some directed-complete strongly Archimedean homogeneous\footnote{
    Homogeneous in this context means that the group of order isomorphisms of the space acts transitively on the interior of the positive cone. Homogeneous ordered spaces have several nice properties~\cite{vinberg1967theory} and are closely related to Euclidean Jordan algebras by the Koecher-Vinberg theorem~\cite{koecher1957positivitatsbereiche}.
    }
    ordered vector space~\cite{wetering2018characterisation}).

The purely a-convex normal SEAs do not seem to have an easily visible structure or classification however.
For instance, we can take an a-convex normal SEA and 
take the horizontal sum with itself to get a new a-convex normal SEA. 
We can then take any amount of such a-convex horizontal sums 
and take their direct sum. Any collection of such direct sums can then
again be combined into a horizontal sum. In this way we can create 
arbitrarily deeply nested sequential effect algebras.

The following proposition shows that much of the structure of a-convex normal SEAs reduces to that of convex normal SEAs:

\begin{proposition}
	Let $E$ be an a-convex normal SEA. Then $E$ can be written as the (possibly non-disjoint) union of convex normal sub-SEAs.
\end{proposition}
\begin{proof}
	Let $a\in E$ be arbitrary. We will construct a convex normal sub-SEA that contains $a$.
	Consider the sub-algebra $\{\frac12\cdot a\}''$. Let $q_c$ be the idempotent that defines its convex part. Obviously $\frac12\cdot a$ belongs to the convex part and hence $\frac12\cdot a\leq q_c$ so that $\frac12\cdot a$ is orthogonal to $q_c^\perp$. 
	But then $\frac12\cdot a$ is also orthogonal to $\frac12\cdot q_c^\perp$ so that they commute. 
	Since there is necessarily an element $b\in \{\frac12\cdot a\}''$ such that $b\ovee b = q_c$, which is for the same reasons as before orthogonal to $\frac12\cdot q_c^\perp$, this element also commutes with both $a$ and $\frac12\cdot q_c^\perp$. 
	Let $D := \{\frac12\cdot a, \frac12\cdot q_c^\perp, b\}''$. Then $D$ is a normal sub-SEA, and $1=2(\frac12\cdot q_c^\perp \ovee b)$, so that $1$ belongs to the convex part, and hence $D$ is also convex.
\end{proof}

We suggest the following axiom to be added to those already present in a SEA in order to remove the kind of pathology introduced by a-convex algebras.

\begin{definition}
    We say a sequential effect algebra has \Define{commuting halves} when $a\commu b$ and $b=c\ovee c$ implies that $a\commu c$.
\end{definition}

\begin{proposition}\label{prop:a-convexcommutewithhalves}
A normal a-convex SEA~$E$ is convex iff it has commuting halves.
\end{proposition}
\begin{proof}
Suppose that $E$ is convex,
and let $a,b,c \in E$ with $a\commu b$ and $b=c\ovee c$
be given.
We must show that~$a \commu c$.
As halfs in normal convex SEAs are unique (cf.~Theorem~\ref{thm-a-convex-thm}), we have
 $c = b\mult(\frac{1}{2} \cdot 1)$.
As~$\frac{1}{2}\cdot 1$ is central (because $E$ is convex), it commutes with $a$, and since $a$ also commutes with $b$, we see that $a$ commutes with~$b\mult(\frac12 \cdot 1) = c$.

Now suppose that $E$ is a-convex 
and has commuting halves. 
Then~$\frac{1}{2}\cdot 1$ commutes with any~$a\in E$,
because~$a\commu 1$.
  Since~$\frac{1}{2}\cdot 1$ is therefore central,
  $E$ is convex by Theorem~\ref{thm-a-convex-thm}.
\end{proof}

By this proposition we see that for normal SEAs the property of having commuting halves is equivalent to having unique halves, i.e.~that $a\ovee a = b\ovee b$ implies $a=b$. This property might be seen as more operationally meaningful than having commuting halves. The reason we chose our axiom to be about commutation is so that it fits in better with the other axioms of Definition~\ref{defn:sea} that require certain elements to commute.

Combining this proposition with Theorem~\ref{thm:maintheorem} then easily gives the following.

\begin{theorem}\label{thm:commuting-halves}
    Let~$E$ be a normal sequential effect algebra with commuting halves. Then $E= B\oplus E_c$ where $B$ is a complete
    Boolean algebra and $E_c$ is a convex normal SEA.
\end{theorem}

\section{Associative sequential products}\label{sec:assoc}

In this section we will go on a tangent, and explore some of the
consequences 
of our representation theorem for normal SEAs (Theorem~\ref{thm:maintheorem})
when the sequential product is associative.

In~\cite{gudder2001sequential} it was noted that it seems
reasonable to expect the sequential product 
to be associative,
that is,
to satisfy $a\mult(b\mult c) = (a\mult b)\mult c$ 
for all $a,b$ and~$c$
(and not just for commuting~$a$ and~$b$).
This is however not the case in quantum theory, 
where the expression $(a\mult b)\mult c$ 
does not correspond to 
any particular physical quantity when $a$ and $b$ do not commute. 
The sequential product in classical probability theory is, of course,
associative. 
This raises the question of whether associativity 
is somehow connected to classicality
of the sequential product.

One might, for example, surmise that
an associative sequential product is necessarily commutative.
For normal SEAs with commuting halves,
this turns out to be correct, see Proposition~\ref{prop-assoc-comm-halves}.
In the presence
of pure a-convexity, on the other hand,
the connection breaks down:
the non-commutative (normal) SEAs in
Examples~\ref{ex:assoc-SEA} and~\ref{ex:horizontal-interval}
have associative sequential products.

\begin{proposition}\label{prop:assoc-is-commutative}
Let $E$ be an associative SEA, and let $p\in E$ be idempotent. 
Then $p$ is central.
\end{proposition}
\begin{proof}
Let~$a\in E$ be given;
we must show that~$p$ commutes with~$a$.
Since~$p^\perp\mult a\leq p^\perp$ is orthogonal to $p$,
we have $0 = (p^\perp\mult a)\mult p = p^\perp\mult (a\mult p)$---using
  associativity here.
  But then $p^\perp \commu a\mult p$, and hence $p\commu a\mult p$. 
Since similarly $p\commu a\mult p^\perp$,
  we get that $p$ commutes with $a\mult p^\perp \ovee a\mult p = a$. 
\end{proof}

This gives some motivation for the lack of associativity in the
sequential product of quantum theory: if it were associative, while
still satisfying all the other axioms of a SEA, then every sharp
measurement (those consisting of idempotent effects, i.e.~projections) had to be classical, in the sense that it is non-disturbing for any other measurement.

\begin{proposition}
  \label{prop-assoc-comm-halves}
An associative normal SEA $E$ with commuting halves is commutative.
\end{proposition}
\begin{proof}
    By Theorem~\ref{thm:commuting-halves}, $E$ is a direct sum of a Boolean algebra and a convex normal SEA. As the Boolean algebra is always commutative we can without loss of generality assume that $E$ is convex.
    By Proposition~\ref{prop:assoc-is-commutative}, all idempotents
    are central, and hence by Theorem~\ref{thm-a-convex-thm},
    $\lambda\cdot p$ is central for any $\lambda\in[0,1]$ and $p$
    idempotent, and hence any $\bigovee_{i=1}^n \lambda_i \cdot p_i$ with
    the $p_i$ idempotent and orthogonal is central. Call such
    elements \Define{simple}. By the spectral theorem (Corollary~\ref{seaspectral}), any
    element~$a$ can be written as $a=\bigvee_n a_n$ where $a_1\leq a_2\leq\ldots
    $ is an increasing sequence of simple elements,
    and hence central.  Thus~$E=Z(E)$ is commutative.
\end{proof}

At the start of this section it was noted that the horizontal sum of unit intervals Example~\ref{ex:horizontal-interval} is an example of a non-commutative associative normal SEA.
As Proposition~\ref{prop:assoc-a-convex} will show, this is in a sense the most general possible example, when restricting to a-convex factors.

\begin{lemma}
  \label{lem:multfloor}
  Let $E$ be a normal SEA with $a,b\in E$.
  Then~$a\mult b=a$ implies $a\mult \floor{b}=a$.
\end{lemma}
\begin{proof}
Note that~$a\commu b$.
Indeed, since~$a\mult b=a$, we have~$a\mult b^\perp = 0$,
  so~$a\commu b^\perp$, and thus~$a \commu b$.
  Then~$a\mult b^2 = a\mult (b\mult b)=(a\mult b)\mult b = a\mult b= a$.
Applying the same reasoning to~$b^2$ in place of~$b$, we
get~$a\mult b^4=0$.
By induction we get~$a\mult b^{2^n}=0$ for all~$n$.
Since~$b^n\leq b^{2^n}$,
we get~$a\mult b^n=0$ for all~$n$ too,
and so~$a\mult\floor{b}=a\mult\bigwedge_n b^n
  =\bigwedge_n a\mult b^n=\bigwedge_n a =a$.
\end{proof}

\begin{lemma}
    Let $E$ be a normal SEA. If the only idempotents in $E$ are $0$ and $1$, then $a\mult b = 0$ implies $a=0$ or $b=0$.
\end{lemma}
\begin{proof}
Since~$a\mult b=0$,
we have~$a \mult b^\perp = a$,
  and so~$a\mult\floor{b^\perp} = a$ by Lemma~\ref{lem:multfloor}.
Thus~$a=0$ when~$\floor{b^\perp}=0$.
  Otherwise, $\floor{b^\perp}=1$
  (since~$0$ and~$1$ are the only idempotents in~$E$),
  so~$1=\floor{b^\perp}\leq b^\perp \leq 1$,
  which implies that~$b^\perp=1$,
  and thus~$b=0$.
\end{proof}

\begin{proposition}\label{prop:assoc-a-convex}
    Let $E$ be a normal a-convex factor with an associative sequential product. Then there is some index set $I$ such that $E$ is isomorphic to the horizontal sum of $\# I$ copies of the real unit interval $[0,1]$.
\end{proposition}
\begin{proof}
    By Proposition~\ref{prop:assoc-is-commutative} any idempotent is central. 
    However, by assumption $Z(E) = \{0,1\}$
    so that the only idempotents are $0$ and $1$. 
    Hence, by the previous lemma, 
    $E$ does not have any non-trivial zero divisors.%
\footnote{Here a \emph{zero divisor} is an
  element~$a\in E$ for which there is~$b\in E$ with $a\mult b=0$---or (by S3) equivalently, $b\mult a=0$.}

    Let $S\subseteq E$ be any non-empty subset of mutually commuting elements,
    and let $q_c$ be the idempotent
    (from Theorem~\ref{thm:first})
        with~$q_c  \mult S^{\prime\prime}$ convex and
        with~$q_c^\perp \mult S^{\prime\prime}$ Boolean.
 Then $q_c = 1$ or $q_c = 0$. If $q_c = 0$,
    then $S''$ is Boolean and hence $S''=\{0,1\}$. As $S\subseteq
    S''$, the only elements of $S$ can then be $0$ and $1$. 
    But since~$1$ and~$0$ are central we would then have $\{0,1\} = S''=E$,
    so that~$E$ is not a-convex. Hence $q_c = 1$, and so~$S''$ is convex.
    As $E$ has no non-trivial zero divisors, the same is true for $S''$. But then
    $S''$ is a convex directed-complete effect monoid without
    non-trivial zero divisors so that 
    by~\cite[Theorem~71]{first} 
    we have $S''\cong [0,1]$.
    As $S\subseteq S''$ there is then a unique additive 
    unital map~$\varphi_S \colon [0,1] \to E$
        such that for each~$s \in S$,
        there is a~$\lambda_s \in [0,1]$
        with~$\varphi_S(\lambda_s) = s$.
    If~$T \supseteq S$ is a larger set of mutually commuting
        elements, and there is some~$s \in S$ with~$s \notin \{0,1\}$,
            then it follows from part~\ref{prop:add-into-nsea-4}
                of Proposition~\ref{prop:add-into-nsea}
                that~$\varphi_S = \varphi_T$.

    Let $I$ be a maximal collection of non-commuting elements
        of~$E \setminus \{0,1\}$ (which exists by Zorn's lemma).
    We claim that $E$ is
    isomorphic to the horizontal sum of $\# I$ copies of $[0,1]$,
        where the sequential product is defined by~$(\lambda, a) \mult (\mu, b) := (\lambda \mu, a)$.
    To show this, we will construct an isomorphism~$\Theta\colon F := \HS([0,1]_{a \in I}, I) \to E$,
    where $\HS([0,1]_{a\in I}, I)$ is as in 
    Definition~\ref{def:horizontal-sum}.
    For~$a \in I$, write~$\varphi_a := \varphi_{\{a\}}$.
    Define~$\Theta(\lambda, a) := \varphi_a (\lambda)$.
    Note that for $a,b \in I$ with~$a \neq b$ 
    we have~$\varphi_a(\lambda) \neq \varphi_b(\mu)$ for all $\lambda,\mu\in(0,1)$
    and hence $\Theta$ is injective.
    
    It is easy to see that~$\Theta$ is unital and additive.
    It then remains to show that $\Theta$ is multiplicative, surjective, and an order embedding, i.e.~that $\Theta(\lambda,a)\leq \Theta(\mu,b)$ implies $(\lambda,a)\leq (\mu,b)$.

To show it is multiplicative, suppose~$(\lambda,a), (\mu,b) \in F$ 
  are given with $\lambda,\mu \in (0,1)$. 
We need to show that $\varphi_a(\lambda)\mult \varphi_b(\mu) 
  = \varphi_a(\lambda\mu)$.
For the moment assume that~$\mu = 2^{-n}$ for some~$n \in \N$.
    Note that~$2^n ( \varphi_a(\lambda) \mult \varphi_b(2^{-n}) )
        = \varphi_a(\lambda)$ by additivity
            and so~$\varphi_a(\lambda) \mult \varphi_b(2^{-n})$ commutes
                with~$\varphi_a(\lambda)$.
        Thus by the earlier
        analysis~$\varphi_a(\lambda) \mult \varphi_b(2^{-n}) = \varphi_a(\zeta)$
            for some~$\zeta$.
            As~$2^n \varphi_a(\zeta) = \varphi_a(\lambda)$
                we must have~$\zeta = 2^{-n}\lambda$,
                and hence
                \[\Theta(\lambda, a) \mult
                            \Theta(2^{-n}, b)
                    := \varphi_a(\lambda) \mult \varphi_b(2^{-n})
                    = \varphi_a(\lambda 2^{-n})
                    =: \Theta((\lambda, a) \mult (2^{-n}, b)).\]
            Using additivity of $\Theta$ we can show the same result for $\mu = m2^{-n}$ and by taking joins we get the result for arbitrary $\mu$ so that indeed~$\Theta(\lambda, a) \mult \Theta(\mu, b)
                    = \Theta((\lambda, a) \mult (\mu, b))$
                for arbitrary~$\mu$.
                Thus~$\Theta$ is multiplicative.
        
To show~$\Theta$ is surjective, suppose~$a \in E$ is given.
As~$\Theta$ is additive and unital, the cases~$a \in \{0,1\}$
    are already covered, so assume~$a \notin \{0,1\}$.
There must be some~$b \in I$ with which~$a$ commutes,
    for otherwise~$I$ would not be maximal.
But then~$a = \varphi_b(\lambda)
        = \Theta(\lambda, b)$ for some~$\lambda$ as desired.

To show~$\Theta$ is an isomorphism, it remains to be shown
    that~$\Theta$ is an embedding, and as $x\leq y \iff x\perp y^\perp$
    for any $x,y$ in an effect algebra
    it hence suffices to show that for any
    $(\lambda,a),(\mu, b) \in F$
    with~$\Theta(\lambda, a) \perp \Theta(\mu, b)$ we have 
    $(\lambda,a) \perp (\mu, b)$.

So let~$(\lambda,a),(\mu, b) \in F$
    be given with~$\Theta(\lambda, a) \perp \Theta(\mu, b)$,
    so that the sum $\varphi_a(\lambda)\ovee \varphi_b(\mu)$ exists.
    When $\lambda\in \{0,1\}$, the desired result is trivial, so assume $\lambda,\mu \notin \{0,1\}$.
    Then in order to show that $(\lambda,a) \perp (\mu, b)$ 
    we need $a=b$ and $\lambda+\mu \leq 1$.

By surjectivity there is some~$c \in I$ and~$\zeta \in [0,1]$
    with~$\varphi_a(\lambda) \ovee \varphi_b(\mu) = \varphi_c(\zeta)$.
By multiplicativity of $\Theta$ we have for any~$\alpha \in (0,1)$:
    \[
         \varphi_c(\alpha \lambda) \ovee
          \varphi_c(\alpha \mu)
          =
    \varphi_c(\alpha) \mult (\varphi_a(\lambda) \ovee \varphi_b(\mu))
    = \varphi_c(\alpha) \mult \varphi_c (\zeta)
    = \varphi_c(\alpha\zeta).\]
    But then~$\alpha\lambda + \alpha\mu  = \alpha\zeta$ 
    and hence~$\alpha\lambda + \alpha\mu \leq 1$.
As~$\alpha \in (0,1)$ was arbitrary,
    we get~$\lambda + \mu = \zeta$ and $\lambda+\mu\leq 1$.
It remains to show that $a=b$.
Let $n\in \N$ be such that $2^{-n} \leq \lambda, \mu$.
We then also have $\varphi_a(2^{-n})\perp \varphi_b(2^{-n})$ 
and hence by the previous argument there is a $c$ such that
$\varphi_a(2^{-n})\ovee \varphi_b(2^{-n}) = \varphi_c(2^{1-n})$.
Adding these sums to themselves $2^{n-1}$ times and using additivity
of the $\varphi$'s we then get 
$\varphi_a(\frac12) \ovee \varphi_b(\frac12) = \varphi_c(\frac12 \ovee \frac12) = 1$.
Hence $\varphi_b(\frac12) = \varphi_a(\frac12)^\perp$ so that $\varphi_a(\frac12)$
and $\varphi_b(\frac12)$ commute. We must then have $a=b$.
\qedhere
\end{proof}

\section{Conclusions and outlook}\label{sec:conclusion}

We have shown that any normal sequential effect algebra can be decomposed into a direct sum of a complete Boolean algebra, a convex normal sequential effect algebra and a new kind of normal sequential effect algebra that we have dubbed `purely a-convex'.

Our proofs relied in a few crucial ways on the assumption of directed completeness of the sequential effect algebra. It would be interesting to see
if a similar characterization or decomposition is possible 
in the presence of just $\omega$-completeness,
that is,
for what Gudder and Greechie dubbed 
$\sigma$-SEAs~\cite{gudder2002sequential}. 
The characterisation of $\omega$-complete effect monoids 
of~\cite{first} imposes considerable restrictions
on the commutative subalgebras of such a 
$\sigma$-SEA, but it is not exactly clear how these can be brought
to bear on the global structure of the $\sigma$-SEA.

We introduced the notion of an a-convex factor for those a-convex algebras where the only central elements are $0$ and $1$. Examples of a-convex factors include any horizontal sum constructed out of a-convex SEAs. We do not know of any other examples of a-convex factors. This might point towards a characterisation of a-convex factors as horizontal sums of a-convex SEAs.

Many results and structures originally proven or defined for operators on Hilbert spaces generalise to C$^*$-algebras, and more generally to Jordan algebras, which were originally defined to model the space of self-adjoint operators on a Hilbert space, but also allow for, for instance, quaternionic Hilbert spaces~\cite{jordan1933}.
It is then not surprising that while the sequential product operation was originally defined for Hilbert space operators, it also works for \emph{JB-algebras}~\cite{hanche1984jordan}, the `Jordan equivalent' of a C$^*$-algebra~\cite{wetering2019commutativity}.
In fact, the only known examples of convex normal SEAs are unit intervals of JB-algebras.
This raises the question whether \emph{all} convex normal SEAs come from JB-algebras in this way.

What is known is that convex normal SEAs are order-isomorphic to the unit interval of a complete strongly Archimedean ordered vector space that furthermore has a homogeneous positive cone~\cite{wetering2018characterisation}, which already greatly restricts the possibilities for such algebras.
Furthermore, when the space is finite-dimensional, and the sequential product is norm-continuous in the first argument, the space must indeed be a JB-algebra~\cite{wetering2018sequential}. What is also known is that when the normal sequential product satisfies for all idempotents $p$ and $q$ the identity $(p\mult q)^2 = p\mult (q\mult p)$ and the implication $\omega(q) = 1 \implies \omega(q\mult p) = \omega(p)$ for any state $\omega$, then the space must be a JB-algebra~\cite{wetering2018sequential}.
There are many elegant characterisations of Jordan algebras~\cite{koecher1957positivitatsbereiche,alfsen2012geometry,barnum2019strongly}
that in turn can be used to characterise quantum theory~\cite{barnum2014higher,wetering2018reconstruction}.
If the existence of a sequential product also characterises Jordan algebras this would give another pathway to understanding the fundamental structures of quantum mechanics.
Indeed, combined with the results of this paper this would give a way to reconstruct quantum theory
without even referring a priori to the concept of real numbers or probabilities.

\medskip 
\noindent\textbf{Acknowledgements}: JvdW is supported by a Rubicon fellowship financed by the Dutch Research Council (NWO). The authors wish to thank the anonymous reviewers for their valued input which has led to significant improvements to the presentation of the paper.

\bibliographystyle{plainnat}
\bibliography{main}
\end{document}